\documentclass{article}
\usepackage[utf8]{inputenc}

\usepackage{amsmath}
\usepackage{comment}
\usepackage{amsfonts}
\usepackage{amsthm}

\newtheorem{theorem}{Theorem}{\bfseries}{\normalfont}
\newtheorem{proposition}{Proposition}{\bfseries}{\normalfont}
\newtheorem{lemma}{Lemma}{\bfseries}{\normalfont}
{\bfseries}{\normalfont}
\newtheorem{corollary}{Corollary}{\bfseries}{\normalfont}
 \newtheorem{obs}{Observation}{\bfseries}{\normalfont}

\usepackage{tikz}
\usetikzlibrary{positioning,backgrounds,patterns,calc}

\usepackage{bbding}
\usepackage{placeins}

\DeclareMathOperator{\dist}{dist}
\DeclareMathOperator{\conf}{conf}

\newcommand{\decprob}[3]{%
\vspace{3pt}
\begin{center}%
    \begin{minipage}{0.9\linewidth}%
      \textsc{#1}\\
      \textbf{Input:} #2\\
      \textbf{Question:} #3
    \end{minipage}%
  \end{center}%
  \vspace{3pt}
}

\usepackage{microtype,ellipsis}

\usepackage[pdfdisplaydoctitle,pagecolor=orange!40!black,menucolor=orange!40!black,filecolor=magenta!40!black,urlcolor=blue!40!black,linkcolor=red!40!black,citecolor=green!40!black,colorlinks]{hyperref}

\hypersetup{pdftitle={On Structural Parameterizations
for the 2-Club Problem}, pdfauthor={Sepp Hartung, Christian Komusiewicz, André Nichterlein, and  Ondrej Suchý}}

\usepackage{units}

\usepackage[numbers,sort,square]{natbib}
\makeatletter
\def\NAT@spacechar{~}% NEW
\makeatother

\usepackage{cleveref}
\crefname{rrule}{Rule}{Rules}
\usepackage{enumitem}

\newcommand{\sclub}{\textsc{$s$-Club}\xspace}
\newcommand{\tclub}{\mbox{\textsc{2-Club}}\xspace}

\newcommand{\Clique}{\textsc{Clique}\xspace}
\newcommand{\MCC}{\textsc{Multicolored Clique}\xspace}

\renewcommand{\c}{\operatorname{c}}

\DeclareRobustCommand{\NoKernelAssume}{$\text{NP}\subseteq \text{{coNP/poly}}$}

\newcommand{\C}{\mathcal{C}}

\newcommand{\size}{\ell}

\graphicspath{{images/}}

\begin{document}

% \begin{frontmatter}

% \title{On Structural Parameterizations
%  for the 2-Club Problem\tnoteref{label1}}
%  \tnotetext[label1]{An extended abstract of this paper appeared in \emph{Proceedings of the 39th International Conference on Current Trends in Theory and Practice of Computer Science (SOFSEM'13)},
%  Jan. 2013, volume~7741 of LNCS, pages 233-243, Springer, 2013~\cite{HKN12b}.}
% %  An important difference compared to the conference version is the fixed-parameter algorithm for the parameter \emph{distance to cographs}.}
% 
% \author[inst1]{Sepp Hartung}             \ead{sepp.hartung@tu-berlin.de}
% \author[inst1]{Christian Komusiewicz}          \ead{christian.komusiewicz@tu-berlin.de}
% \author[inst1]{André Nichterlein}          \ead{andre.nichterlein@tu-berlin.de}
% \author[inst2]{Ondřej Suchý} \ead{ondrej.suchy@fit.cvut.cz}
% 
% \address[inst1]{Institut f\"ur Softwaretechnik und Theoretische Informatik,
%   TU Berlin,
%   D-10587~Berlin, Germany}
%   
% \address[inst2]{Faculty of Information Technology, Czech Technical University in Prague,
%   Prague, Czech Republic}
  \title{On Structural Parameterizations for the 2-Club Problem\footnote{An extended abstract of this paper appeared in \emph{Proceedings of the 39th International Conference on Current Trends in Theory and Practice of Computer Science (SOFSEM'13)},
 Jan. 2013, volume~7741 of LNCS, pages 233-243, Springer, 2013~\cite{HKN12b}.}}
  \author{Sepp Hartung \thanks{TU Berlin, Germany, sepp.hartung@tu-berlin.de}\and Christian Komusiewicz \thanks{TU Berlin, Germany, christian.komusiewicz@tu-berlin.de} \and André Nichterlein \thanks{TU Berlin, Germany, andre.nichterlein@tu-berlin.de} \and Ondřej Suchý \thanks{Czech Technical University in Prague, Czech Republic, ondrej.suchy@fit.cvut.cz}}
  
\maketitle
\begin{abstract}
% \begin{itemize}
%  \item NP-hardness for Clique Cover number 3/dom set 3
% \item NP-hardness for distance to bipartite
% \item w[1]-hardness for $h$-index 
% \item xp-algorithm for $h$-index
% \end{itemize}
 The NP-hard \tclub problem is, given an undirected graph $G=(V,E)$ and~$\size\in \mathbb{N}$, to decide whether there is a vertex set~$S\subseteq V$ of size at least~$\size$ such that the induced subgraph~$G[S]$ has diameter at most two.  We make progress towards a systematic classification of the complexity of \tclub with respect to a hierarchy of prominent structural graph parameters.  First, we present the following tight NP-hardness results: \tclub is NP-hard on graphs that become bipartite by deleting one vertex, on graphs that can be covered by three cliques, and on graphs with \emph{domination number} two and \emph{diameter} three. 
 Then, we consider the parameter \emph{$h$-index} of the input graph. This parameter is motivated by real-world instances and the fact that \tclub is fixed-parameter tractable with respect to the larger parameter \emph{maximum degree}.  We present an algorithm that solves \tclub in~$|V|^{f(k)}$ time with~$k$ being the \emph{$h$-index}. By showing W[1]-hardness for this parameter, we provide evidence that the above algorithm cannot be improved to a fixed-parameter algorithm.  Furthermore, the reduction used for this hardness result can be modified to show that \tclub is NP-hard if the input graph has constant \emph{degeneracy}.
%  This also implies $W[1]$-hardness with respect to the degeneracy of the input graph.
 Finally, we show that \tclub is fixed-parameter tractable with respect to \emph{distance~to~cographs}.

% 
% for the parameter $h$-
% 
%   Given an undirected graph~$G=(V,E)$ and an integer~$\size \ge 1$, the NP-hard~\tclub{} problem asks for a vertex set~$S\subseteq V$ of size at least~$\size$ such that the subgraph induced by~$S$ has diameter at most~$2$.  We provide a multivariate complexity study of~\tclub{} by considering different structural parameterizations for the input graph~$G$. We show that~\tclub~is in~XP  and W[1]-hard for the parameter \emph{$h$-index} of~$G$.  
\end{abstract}
% 
% \begin{keyword}
% clique relaxations\sep cohesive subnetworks\sep fixed-parameter tractability \sep social network analysis
% %% MSC codes here, in the form: \MSC code \sep code
% %% or \MSC[2008] code \sep code (2000 is the default)
% \end{keyword}

% \end{frontmatter}

\section{Introduction}
The identification of cohesive subnetworks is an important task in the analysis of social and biological networks, since these subnetworks are likely to represent communities or functional subnetworks within the large network. The natural cohesiveness requirement is to demand that the subnetwork is a complete graph, a clique. However, this requirement is often too restrictive and thus relaxed definitions of cohesive graphs such as $s$-cliques~\cite{Alb73}, $s$-plexes~\cite{SF78}, and $s$-clubs~\cite{Mok79} have been proposed. In this work, we study the problem of finding large~$s$-clubs within the input network. An~$s$-club is a vertex set that induces a subgraph of diameter at most~$s$. Thus, $s$-clubs are distance-based relaxations of cliques, which are exactly the graphs of diameter one. For constant~$s\ge 1$, the problem of finding~$s$-clubs is defined as follows. 
\decprob{\sclub{}}%
{An undirected graph~$G = (V,E)$ and $\size \in \mathbb{N}$.}  {Is
  there a vertex set~$S \subseteq V$ of size at least~$\size$ such
  that~$G[S]$ has diameter at most~$s$?}
In this work, we study the computational complexity of \tclub{}, that is, the special case of~$s=2$. 
This is motivated by the following two considerations. First, \tclub is an important special case concerning the applications: For biological networks, 2-clubs and 3-clubs have been identified as the most reasonable diameter-based relaxations of cliques~\cite{Pasu08}. Further, \citet{BBT05} also proposed to compute 2-clubs and 3-clubs for analyzing protein interaction networks.
% \todo{falk/christian paper erwähnen, da diese strukturen auch mindestens 2-clubs sind?}
\tclub also has applications in the analysis of social networks~\cite{ML06}. Consequently, the extensive experimental study concentrate on finding 2-~and 3-clubs~\cite{AC12,AC11,BBT05,BLP02,BM11,CHLS11, HKN12}. Second, \tclub is the most basic variant of \sclub that is different from the \textsc{Clique} problem which is equivalent to~\textsc{1-Club}. For example, being a clique is a hereditary graph property, that is, it is closed under vertex deletion. In contrast, being a 2-club is not hereditary, since deleting vertices can increase the diameter of a graph. Hence, it is interesting to spot differences in the computational complexity of the two problems. 

In the spirit of multivariate algorithmics~\cite{Fel09,Nie10}, we aim to describe how structural properties of the input graph determine the computational complexity of~\tclub. We want to determine sharp boundaries between tractable and intractable special cases of \tclub, and whether some graph properties, especially those motivated by the structure of social and biological networks, can be exploited algorithmically. By arranging the parameters in a hierarchy (ranging from large to small parameters) we draw a border line between tractability and intractability to obtain a systematic view on ``stronger parameterizations'' (refer to~\cite{KN12} for further discussion of the parameter hierarchy and its application). A similar approach was followed for other hard graph problems such as \textsc{Odd Cycle Transversal}~\cite{JK11} and for the computation of the pathwidth of a graph~\cite{BJK12}.

The structural properties that we consider, called structural graph parameters, 
are usually described by integers;
well-known examples of such parameters are the maximum degree or the
treewidth of a graph.
% we provide an overview of the complexity of~\tclub{} with respect to
% structural parameters of the input graph that measure the distance
% to different graph classes.
Our results use the classical framework of
NP-hardness as well as the framework of parameterized complexity
to show (parameterized) tractability and intractability of \tclub with respect to the structural graph parameters under consideration.
That is, for some graph parameters we show that \tclub becomes NP-hard in case of constant parameter values, whereas for other graph parameters we show fixed-parameter (in)tractability.

\subsection{ Related Work}
For all~$s\ge 1$,~\sclub is NP-complete on graphs of
diameter~$s+1$~\cite{BBT05}; \textsc{2-Club} is NP-complete even on
split graphs and, thus, also on chordal graphs~\cite{BBT05}.\footnote{\looseness=-1 An NP-hardness reduction given by \citet[Theorem~1]{BBT05} can be easily modified such that the  \tclub instance is a split graph (make the vertex set~$E$ a clique).}
In contrast,~\tclub~is solvable in polynomial time on bipartite graphs, on trees, and on interval graphs~\cite{Sch09}.  \citet{GHKR13} consider the complexity of \sclub in special graph classes. For instance they prove polynomial-time solvability of \sclub on choral bipartite, strongly chordal and distance hereditary graphs. Additionally, it is proven that on a superclass of these graph classes, called weakly chordal graphs, it is polynomial time solvable for odd~$s$ and NP-hard for even~$s$.
% Applications of~\sclub{} include prediction of protein complexes in protein-interaction networks~\cite{Pasu08} and social network analysis~\cite{ML06}. Thereby, it turned out that 2-clubs and 3-clubs seemed to be the most interesting cases on practical instances~\cite{Pasu08,BM11}.

 \sclub  is well-understood from the viewpoint of approximation algorithms~\cite{AMS10}: 
%  \sclub cannot be approximated within factor $n^{\nicefrac{1}{3}-\epsilon}$ for any $\epsilon>0$, unless $\text{P}=\text{NP}$~\cite{MM02}; this bound has been
 It is NP-hard to approximate \sclub within a factor of $n^{\frac{1}{2}-\epsilon}$ for any $\epsilon>0$. On the positive side, it has been shown that a largest set consisting of a vertex together with all vertices within distance~$\left\lfloor \frac{s}{2}\right\rfloor$ is a factor $n^{\frac{1}{2}}$ approximation for even~$s\ge 2$ and a factor $n^{\frac{2}{3}}$ approximation for odd~$s\ge 3$.  
Several heuristics~\cite{BLP00,CHLS11,AC11, CHLS11}, integer linear programming formulations~\cite{BLP02,BBT05,AC12}, fixed-parameter algorithms~\cite{SKMN11,HKN12}, and branch-and-bound algorithms~\cite{BLP02} have been proposed and experimentally evaluated~\cite{BM11,HKN12}.
% A first algorithmic study of \sclub was undertaken by~\citet{BLP00}
% who presented some heuristics.  \citet{BLP02} proposed and evaluated
% an integer linear program based algorithm as well as a
% branch-and-bound algorithm. \citet{BBT05} give formulations as integer
% linear program of~\sclub and the related~\textsc{$s$-Clique}
% problem. See also \citet{BM11} for more recent experimental studies.
% 
% \begin{obs}\label{obs:moore-bound}
% For a graph~$G=(V,E)$ with diameter~$s$ and maximum degree~$\Delta$ it holds that
% \begin{equation*}
%  |V|\le 1+ \Delta\sum_{i=0}^{s-1}(\Delta-1)^i\text{.}
% \end{equation*}
% \end{obs}
% % 
% Note that graphs that fit the upper bound given by \autoref{obs:moore-bound} are called \emph{Moore graphs}. Moreover, it follows that a diameter-two graph with~$\size$ vertices has a vertex of degree at least $\sqrt{\size-1}$. Hence, denoting the maximum degree in a graph by~$\Delta$, for non-trivial instances of \tclub it follows that $\sqrt{\size-1}\le \Delta< \size -1$. 
 
From the viewpoint of parameterized algorithmics, \textsc{1-Club} is equivalent to \textsc{Clique} and thus W[1]-hard with respect to~$\size$~\cite{DF99}. 
In contrast, for all~$s\ge 2$, \sclub is fixed-parameter tractable with respect to~$\size$~\cite{SKMN11,CHLS11} 
and also with respect to the parameter treewidth of~$G$~\cite{Sch09}.
Additionally, a search tree-based algorithm that branches into the two possibilities to delete one of two vertices with distance more than~$s$ achieves a running time of $O(2^{n-\size}\cdot nm)$ for the \emph{dual parameter} $n-\size$ which measures the \emph{distance to a $s$-club}~\cite{SKMN11}.\footnote{\looseness=-1 \citet{SKMN11} considered finding an $s$-club of size \emph{exactly}~$\size$. The claimed fixed-parameter tractability with respect to~$n-\size$ however only holds for the problem of finding an~$s$-club~of size \emph{at least}~$\size$. The other fixed-parameter tractability results hold for both variants.}
This algorithm cannot be improved to $O((2-\epsilon)^{n-\size}\cdot n^{O(1)})$ for any~$\epsilon>0$ unless the strong exponential time hypothesis fails~\cite{HKN12}. Interestingly, \citet{CHLS11} proved that with respect to the number of vertices~$n$ the same search tree algorithm runs in $O(1.62^n)$ time.

The main observation behind the fixed-parameter algorithm for~$\size$ is that any closed neighborhood~$N[v]$ of a vertex~$v$ is an $s$-club for $s\ge 2$. Hence, the maximum degree~$\Delta$ in non-trivial instances is less than~$\size-1$. It also holds, however, that~$\size\le \Delta^s+1$ in yes-instances. Thus, for constant~$s$, fixed-parameter tractability with respect to~$\size$ also implies is fixed-parameter tractability with respect to the maximum degree of~$G$.
Moreover, \sclub does not admit a polynomial kernel with respect to~$\size$ (unless \NoKernelAssume)~\cite{SKMN11}. Interestingly, taking for each vertex the vertex itself together with all other vertices that are in distance at most~$s$ forms a so-called Turing-kernel with at most $k^2$-vertices for even~$s$ and at most $k^3$-vertices for odd~$s$~\cite{SKMN11}. In companion work~\cite{HKN12}, we considered different structural parameters: We presented a fixed-parameter algorithm for the parameter treewidth and polynomial kernels for the parameters (size of a)  \emph{feedback edge set} and the \emph{cluster editing number}. Additionally, we showed the non-existence of a polynomial kernel and that the simple search tree-based algorithm for the \emph{dual parameter} $n-\size$ is asymptotically optimal. Somewhat in contrast to this negative result, we showed that an implementation of the branching algorithm for the \emph{dual parameter} combined with the Turing-kernelization is among the best-performing algorithms on real-world  
and on synthetic instances.

\subsection{Structural Parameters}
\looseness=-1 We next define the structural parameters under consideration (see \autoref{fig:result-overview} for an illustration of their relations).
For a set of graphs~$\Pi$ (for instance, the set of bipartite graphs) the parameter \emph{distance to~$\Pi$} measures the number of vertices that have to be deleted in the input graph in order to obtain a graph that is isomorphic to one in~$\Pi$.
Denoting by~$P_t$ an induced path on~$t$ vertices, the set of $P_t$-free graphs consists of all graphs not containing any~$P_t$. The $P_4$-free graphs are called cographs, $P_3$-free graphs are so-called cluster graphs, and connected $P_3$-free graphs are cliques. A graph where each connected component is an $s$-club is called \emph{$s$-club cluster graph}. 
Observe that this is equivalent to requiring that all shortest paths do not contain any~$P_{s+2}$.
% A graph is a \emph{chordal graph} if it does not contain any induced cycle of size at least four.
Additionally, deleting all the vertices on a~$P_t$ is a factor-$t$ approximation for the parameter \emph{distance to $P_t$-free graphs} (or even restricted to $P_t$'s on shortest paths) and, thus, we may assume that such a vertex deletion set is provided as an additional input for the corresponding algorithms.

A graph is a \emph{co-cluster graph} if % the graph that results from deleting all existing edges and inserting all missing edges 
its complement graph 
is a cluster graph.
The \emph{minimum clique cover} is the minimum number of cliques in a graph that are needed to cover all vertices, that is, each vertex is contained in at least one of these cliques. The \emph{domination number} of a graph is the minimum size of a dominating set, this is, a set such that each vertex is contained in it or has at least one neighbor in it.
A \emph{vertex cover} of~$G$ is a vertex set whose deletion transforms~$G$ in a graph without any edges. 
An \emph{independent set} is the complement of a vertex cover.
A set of edge insertions and deletions is a \emph{cluster editing} set if it transforms~$G$ into a cluster graph.
A set of edges is a \emph{feedback edge set} if its deletion results in a graph without any cycle.
A graph has \emph{$h$-index}~$k$, if~$k$ is the largest number such that the graph has at least $k$~vertices of degree at least~$k$. The \emph{degeneracy} of a graph is the smallest number~$d$ such that each subgraph has at least one vertex of degree at most~$d$.
The \emph{bandwidth} of a graph~$G=(V,E)$ is the minimum $k\in \mathbb{N}$ such that there is a function $f:V\rightarrow\mathbb{N}$ with $|f(v)-f(u)| \le k$ for all edges~$\{u,v\} \in E$.

\subsection{Our Contribution}\label{sec:contri}

We make progress towards a systematic classification of the complexity of \tclub with respect to structural graph parameters. 
\autoref{fig:result-overview} gives an overview of our results and their implications. 
\begin{figure}[t]%
\tikzstyle{para}=[rectangle, minimum height=.8cm,fill=gray!20,rounded corners=3mm]
\tikzstyle{fpt}=[rectangle, minimum height=.8cm,fill=white!70,rounded corners=3mm,draw]
\tikzstyle{fptF}=[rectangle, minimum height=.8cm,fill=white!20,rounded corners=3mm,draw]
\tikzstyle{NP}=[rectangle, minimum height=.8cm,fill=white!60,rounded corners=3mm,draw]
\tikzstyle{NPF}=[rectangle, minimum height=.8cm,fill=white!30,rounded corners=3mm,draw]
\tikzstyle{w1}=[rectangle, minimum height=.8cm,fill=white!60,rounded corners=3mm,draw]
\tikzstyle{w1F}=[rectangle, minimum height=.8cm,fill=white!30,rounded corners=3mm,draw]
\tikzstyle{our}=[draw] % , ultra thick
\newcommand{\tworows}[2]{\begin{tabular}{c}{#1}\\{#2}\end{tabular}}
\newcommand{\threerows}[3]{\begin{tabular}{r}{#1}\\{#2}\\{#3}\end{tabular}}
\newcommand{\distto}[1]{\tworows{Distance to}{#1}}
\def\dist{2mm}
\def\tres{10}
\resizebox{\textwidth}{!}{
	\begin{tikzpicture}[node distance=7mm]
		\node[fpt,our] (vc) {Vertex Cover \cite{HKN12}};
		\node[fpt,our] (ce) [right=of vc,xshift=-3mm] {Cluster Editing \cite{HKN12}};
		\node[fptF] (ml) [right=of ce,xshift=-3mm] {Max Leaf \#};
		\node[fptF] (dc) [left=of vc,fill=green!40,xshift=-1.2cm] {Distance to Clique};
 
		\node[fpt] (d2c) [below left=of dc,xshift=2.4cm,yshift=-0.2cm] {\distto{2-club \cite{SKMN11}}} edge[thick] (dc);
 		\node[NP,our] (mcc) [left=of d2c,xshift=2mm] {\tworows{Minimum}{Clique Cover}\FiveStar} edge[thick] (dc);
 		\node[fpt,our] (dcc) [below= of dc,xshift=1.9cm] {\distto{Co-Cluster}\FiveStar} edge[thick] (dc) edge[thick] (vc);
 		\node[fpt,our] (dcl) [below = of vc] {\distto{Cluster}\FiveStar} edge[thick] (dc) edge[thick] (vc) edge[thick] (ce);
 		\node[fptF] (ddp) [below=of ce,xshift=-2mm] {\distto{Disjoint Paths}} edge[thick] (vc) edge[thick] (ml);
 		\node[fpt,our] (fes) [below =of ml] {\tworows{Feedback}{Edge Set \cite{HKN12}}} edge[thick] (ml);
 		\node[fptF,our] (bw) [below right=of ml, xshift=2.2mm, yshift=-2mm] {Bandwidth} edge[thick] (ml);
 
  		\node[NPF] (is) [below=of mcc] {\tworows{Maximum}{Independent Set}} edge[thick] (mcc);
 		\node[fptF] (dcg) [below= of dcc] {\distto{Cograph \FiveStar}} edge[thick] (dcc) edge[thick] (dcl);
 		\node[para] (dig) [below = of dcl] {\distto{Interval}} edge[thick] (dcl) edge[thick] (ddp);
 		\node[fptF] (fvs) [below= of ddp] {\tworows{Feedback}{Vertex Set}} edge[thick] (ddp) edge[thick] (fes);
 		\node[fptF] (pw) [below = of fes] {Pathwidth} edge[thick] (ddp) edge[thick] (bw);

 		\node[NP,our] (ds) [below=of is] {\tworows{Domination}{Number}\FiveStar} edge[thick] (is);
		\node[para] (d2cc) [below=of d2c, yshift=-18mm] {\distto{2-club cluster}} edge[thick] (d2c) edge[thick] (dcg);

 		\node[NP, xshift=17mm,our] (dch) [below= of dcg] {\distto{Chordal} \cite{BBT05}} edge[thick] (dig);
		\node[NP, xshift=-15mm,our] (dbp) [below= of fvs] {\distto{Bipartite}\FiveStar} edge[thick] (fvs);1
 		\node[fpt] (mxd) [below= of bw] {\tworows{Maximum}{Degree~\cite{SKMN11}}} edge[thick] (bw);
%		\node[NPF] (mxdia) [below=of ds] {\tworows{Max Diameter}{of Components}} edge[thick] (ds) edge[thick] (dcg);
		\node[NPF] (mxdia) [below=of d2cc] {Diameter} edge[thick] (ds) edge[thick] (d2cc);
		\node[NPF, xshift=-17mm] (dpf) [below= of dch] {\distto{Perfect}} edge[thick] (dch) edge[thick] (dcg) edge[thick] (dbp);
		\node[fpt] (tw) [below= of pw,our] {Treewidth \cite{Sch09}} edge[thick] (fvs) edge[thick] (pw);
 
%		\node[NPF] (mndia) [below=of mxdia] {\tworows{Min Diameter}{of Components}} edge[thick] (mxdia);
		\node[w1,our] (hindex) [below=of mxd,yshift=2mm] {$h$-index\,\FiveStar} edge[thick] (mxd);
 		\node[w1F] (deg) [below=of hindex,yshift=-2mm] {Degeneracy \FiveStar} edge[thick] (tw) edge[thick] (hindex);

%		\node[NPF] (gir) [below= of mndia] {Girth} edge[thick] (mndia);
		\node[NPF] (cn) [below left=of deg,xshift=-5mm] {\tworows{Chromatic}{Number}} edge[thick] (deg) edge[thick] (dbp);
		\node[NP] (mnd) [below=of deg,our,xshift=-5mm] {\tworows{Average}{Degree}\FiveStar} edge[thick] (deg);

%		\node[NPF, xshift=-15mm] (cl) [below=of cn] {\tworows{Maximum}{Clique}} edge[thick] (cn);
%		\node[NPF, xshift=-15mm] (con) [below=of mnd] {\distto{Disconnected}} edge[thick] (mnd);
%		\node[NPF, xshift=-15mm] (don) [below right=of mnd] {\tworows{Domatic}{Number}} edge[thick] (mnd);
		\node [above=of ce,xshift=14mm,yshift=-7mm] {\bf FPT and polynomial-size kernels};
		\node [below=of mxdia,xshift=18mm] {\bf NP-hard with constant parameter values};
		\node [below=of hindex,xshift=-0.5mm,yshift=7mm] {\bf W[1]-hard};
		\node [left=of tw,xshift=1mm,yshift=-5mm] {\bf FPT, };
 		\node [below=of tw,xshift=-7mm,yshift=5mm] {\bf \threerows{but no polynomial-}{size kernel unless}{\NoKernelAssume}};
		\begin{pgfonlayer}{background}
			\draw[line width=0.5pt,rounded corners,top color=orange!15,bottom color=orange!15] 
				($(vc.north west) 				+ (-\dist, \dist+\tres)$) --
				($(vc.north east) 				+ (0.5*\dist, \dist+\tres)$) -- 
				($(ddp.north  -| vc.east) 		+ (0.5*\dist, \dist)$) -- 
				($(ddp.north east) 				+ ( 1.4*\dist, \dist)$) -- 
				($(pw.north  -| ddp.east) 		+ ( 1.4*\dist, \dist)$) --
				($(pw.north  -| bw.west) 		+ (-\dist, \dist)$) --
				($(bw.north  -| bw.west) 		+ (-\dist, \dist)$) --
				($(bw.north  -| deg.east) 		+ (0.5*\dist, \dist)$) --
%				($(mxd.north west)		 		+ (-\dist, \dist)$) --
%				($(mxd.north  -| deg.east)		+ (0.5*\dist, \dist)$) -- 
				($(mxd.south  -| deg.east)		+ (0.5*\dist,-\dist)$) -- 
				($(mxd.south -| pw.east)		+ (2.0*\dist,-\dist)$) -- 
				($(deg.south -| pw.east)		+ (2.0*\dist,-\dist)$) -- 
				($(deg.south -| dbp.east)		+ (1.5*\dist,-\dist)$) -- 
				($(dcg.south -| dbp.east)		+ (1.5*\dist,-\dist)$) -- 
				($(dcg.south -| dig.east)		+ ( \dist,-\dist)$) -- 
				($(dig.north east)				+ ( \dist, \dist)$) -- 
%				($(dig.north -| dcc.west)		+ (-0.7*\dist, \dist)$) --
				($(dig.north west)				+ (-\dist, \dist)$) -- 
				($(dcg.south -| dig.west)		+ (-\dist,-\dist)$) -- 
				($(dcg.south -| dcg.west) 		+ (-0.5*\dist,-\dist)$) --
%				($(d2c.south -| dcg.west) 		+ (-1.4*\dist,-1*\dist)$) --
				($(d2c.south west) 				+ (-0.5*\dist,-1*\dist)$) --
				($(mcc.north -| d2c.west)		+ (-0.5*\dist, \dist)$) --
				($(mcc.north  -| vc.west) 		+ (-\dist, \dist)$) -- cycle;
			\draw[line width=0.5pt,rounded corners,top color=green!15,bottom color=green!15] 
				($(ce.north west) 				+ (-1*\dist, \dist+\tres)$) --
				($(ce.north  -| ml.east) 		+ ( 2*\dist, \dist+\tres)$) -- 
				($(fes.south  -| ml.east) 		+ ( 2*\dist,-2*\dist)$) --
				($(fes.south  -| ml.west) 		+ (-2.2*\dist,-2*\dist)$) --
				($(ce.south  -| ml.west) 		+ (-2.2*\dist,-2*\dist)$) -- 
				($(ce.south west) 				+ (-1*\dist,-2*\dist)$) -- cycle;
			\draw[line width=0.5pt,rounded corners,fill=red!30] 
				($(mcc.north  -| ds.west) 		+ (-1.2*\dist, \dist)$) -- 
				($(mcc.north  -| ds.east) 		+ ( 1.2*\dist, \dist)$) -- 
				($(d2cc.south  -| ds.east) 		+ ( 1.2*\dist,-\dist)$) -- 
				($(d2cc.south east) 					+ ( \dist,-\dist)$) -- 
				($(dbp.north  -| d2cc.east) 		+ ( \dist, \dist)$) -- 
				($(dbp.north  -| dbp.east) 		+ ( \dist, \dist)$) -- 
				($(cn.north  -| dbp.east) 		+ ( \dist, \dist)$) -- 
				($(cn.north  -| deg.west) 		+ (-0.5*\dist, \dist)$) -- 
				($(deg.north west) 				+ (-0.5*\dist, \dist)$) -- 
				($(deg.north east) 				+ (0.5*\dist, \dist)$) -- 
				($(mnd.south  -| deg.east)		+ (0.5*\dist,-\dist)$) -- 
				($(mnd.south -| ds.west) 		+ (-1.2*\dist,-\dist)$) -- cycle;
			\draw[line width=0.5pt,rounded corners,fill=magenta!30] 
				($(hindex.north  -| deg.west) 	+ (-0.5*\dist, \dist)$) -- 
				($(hindex.north  -| deg.east) 	+ (0.5*\dist, \dist)$) -- 
				($(hindex.south  -| deg.east) 	+ (0.5*\dist, -3*\dist)$) -- 
				($(hindex.south  -| deg.west) 	+ (-0.5*\dist, -3*\dist)$) -- cycle;
		\end{pgfonlayer}
	\end{tikzpicture}
}%
\caption{
Overview of the relation between structural graph parameters (see \autoref{sec:contri}) and of our results\,\FiveStar{} for \tclub. An edge from a parameter~$\alpha$ to a parameter~$\beta$ below of $\alpha$ means that~$\beta$ can be upper-bounded in a polynomial (usually linear) function in~$\alpha$. 
% Hence, parameter~$\beta$ is stronger than~$\alpha$~\cite{KN12}.
% The boxes indicate the complexity of \tclub with respect to the enclosed parameters.  
The box containing the parameter (size of a) \emph{vertex cover} on top, consists of all parameters for which \tclub becomes fixed-parameter tractable but does not admit a polynomial kernel~\cite{HKN12}. Therein, for all parameters the best performing algorithms run in $2^{O(2^k)}\cdot n^{O(1)}$ time with the only exception \emph{distance to 2-club} admitting a $2^{k}\cdot n^{O(1)}$-time algorithm~\cite{SKMN11}.
The box consisting of \emph{cluster editing}, \emph{max leaf \#}, and \emph{feedback edge set} contains all parameters admiting a single-exponential time algorithm and a polynomial kernel~\cite{HKN12}. 
The box at the bottom contains those parameters where \tclub remains NP-hard even for constant values.
It is open whether \tclub{} is fixed-parameter tractable when it is parameterized by \emph{distance to interval} or \emph{distance to 2-club cluster} and whether it admits a polynomial kernel when parameterized by \emph{distance to clique}.}
\label{fig:result-overview}
\vspace{-3mm}
\end{figure}
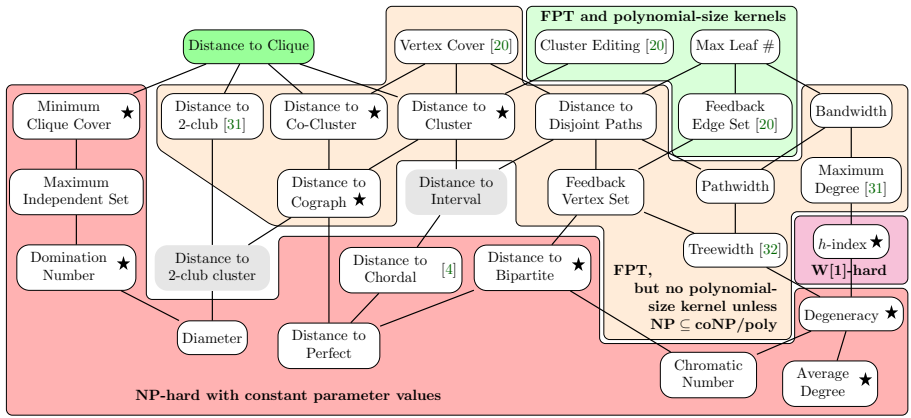
In \autoref{sec:clique-cover}, we consider the graph parameters
\emph{minimum clique cover number}, \emph{domination number}, and
some related graph parameters. We show that \tclub is NP-hard even if
the \emph{minimum clique cover number} of~$G$ is three. In contrast, we show
that if the \emph{minimum clique cover number} is two, then \tclub is
polynomial-time solvable. Then, we show that \tclub is NP-hard even
if~$G$ has a dominating set of size two.
This result is tight in the sense that \tclub is trivially solvable in case~$G$ has a dominating set of size one.
In \autoref{sec:dist-bipartite}, we study the parameter \emph{distance to bipartite graphs}. We show that \tclub is NP-hard even if the input graph can be transformed into a bipartite graph by deleting only one vertex. This is somewhat surprising since \tclub is polynomial-time solvable on bipartite graphs~\cite{Sch09}.
Then, in \autoref{sec:h-index}, we consider the graph
parameter~\emph{$h$-index}. The study of this parameter is motivated by
the fact that the~$h$-index is usually small in social
networks (see \autoref{sec:h-index} for a more detailed
discussion). On the positive side, we show that \tclub is polynomial-time solvable for constant~$k$. On the negative side, we show that
\tclub parameterized by the \emph{$h$-index}~$k$ of the input graph is
W[1]-hard. Hence, a running time of~$f(k)\cdot
n^{O(1)}$ is probably not achievable. Even worse, we prove that \tclub becomes NP-hard even for constant degeneracy. Note that degeneracy is provably at most as large as the $h$-index of a graph.

Finally, in \autoref{sec:fpt-dist-to-cocluster} we  describe a fixed-parameter algorithm for the parameter \emph{distance to cographs} and show that it can be slightly improved for the weaker parameter \emph{distance to cluster graphs}. 
% This implies fixed-parameter tractability for the parameters \emph{distance to cluster and co-cluster graphs}.
% \footnote{A \emph{cluster graph} is a vertex-disjoint union of cliques, and a \emph{co-cluster graph} is the complement graph of a cluster graph.%, that is, it is either an independent set or a complete~$p$-partite graph for some~$p\le n$.}
Interestingly, these are rare examples for structural graph parameters, that are unrelated to treewidth and still admit a fixed-parameter algorithm (see \autoref{fig:result-overview}). Notably, the fixed-parameter algorithm for treewidth and those for \emph{distance to cograph} both have the same running time characteristic, that is, $2^{O(2^k)}\cdot n^{O(1)}$ and this is, so far, also the best for the  much ``weaker'' parameter \emph{vertex cover}.
% are unrelated to the treewidth of~$G$, thus turning them 
% The parameters under
% consideration measure the number~$k$ of vertices that have to be
% deleted in order to obtain either a cluster graph or a co-cluster
% graph. We present fixed-parameter algorithms for both
% parameterizations, thus showing that there are useful graph parameters
% that are unrelated to the 

For the sake of completeness, we would like to mention that for the parameters \emph{bandwidth} and \emph{maximum degree}, 
taking the disjoint union of the input graphs is a composition algorithm that proves the non-existence of polynomial kernels~\cite{BDFH09JCSS}, under the standard assumption that \NoKernelAssume{} does not hold.

\subsection{Preliminaries}\label{sec:prelim}
We only consider undirected and simple graphs~$G=(V,E)$ where~$n:=|V|$ and~$m:=|E|$. For a vertex set~$S\subseteq V$, let~$G[S]$ denote the \emph{subgraph induced by~$S$} and~$G-S:=G[V\setminus S]$. We use~$\dist_G(u,v)$ to denote the \emph{distance between~$u$ and~$v$} in~$G$, that is, the length of a shortest path between~$u$ and~$v$. For a vertex~$v\in V$ and an integer~$t\ge 1$, denote by~$N^G_t(v):=\{u\in V\setminus \{v\}\mid \dist_G (u,v)\le t\}$ the set of vertices within distance at most~$t$ to~$v$. Moreover, we set~$N^G_t[v]:=N^G_t(v)\cup \{v\}$, $N^G[v]:=N^G_1[v]$ and $N^G(v):=N_1(v)$.   If the graph is clear from the context, we omit the superscript~$G$. Two vertices~$v$ and~$w$ are \emph{twins} if $N(v)\setminus\{w\}=N(w)\setminus \{v\}$ and they are \emph{twins with respect to a vertex set~$X$} if $N(v)\cap X=N(w)\cap X$. The twin relation is an equivalence relation; the corresponding equivalence classes are called \emph{twin classes}. The following observation is easy to see and it shows that 
either none or all vertices of a twin class are contained in a maximum-size $s$-club.
\begin{obs} \label{obs:twin-containment}
	Let~$S$ be an $s$-club in a graph~$G = (V,E)$ and let~$u,v \in V$ be twins. 
	If~$u\in S$ and $|S|>1$, then $S\cup \{v\}$ is also an $s$-club in~$G$.
\end{obs}
% 
% For the relevant notions of parameterized complexity we refer to~\cite{DF99,Nie06}. \todo{auch für journal version?}
% comment due to space
% \subsection{ Parameterized Algorithmics}
We briefly recall the relevant notions from parameterized
complexity (see \cite{DF99,FG06,Nie06}). A problem is \emph{fixed-parameter tractable} (FPT) with
respect to a parameter~$k$ if there is a computable function~$f$ such that any instance~$(I,k)$ can be solved in
$f(k)\cdot |I|^{O(1)}$ time.  A problem is contained in XP if it can be solved in~$O(|I|^{f(k)})$ time for some computable function~$f$. A \emph{kernelization algorithm} reduces any instance $(I,k)$ in polynomial time to an equivalent instance~$(I',k')$ with $|I'|, k'\le g(k)$ for some computable~$g$. The instance $(I',k')$ is called \emph{kernel} of size~$g$ and in the special case of~$g$ being a polynomial it is a polynomial kernel.

A problem that is shown to be \emph{W[1]-hard} by means of a \emph{parameterized reduction} from a W[1]-hard problem is not fixed-parameter tractable, unless $\text{FPT}=\text{W[1]}$. A \emph{parameterized reduction} maps an instance $(I,k)$ in $f(k)\cdot |I|^{O(1)}$ time for some function~$f$ 
to an equivalent instance $(I',k')$ with $k'\le g(k)$ for some functions~$f$ and~$g$.

% If $g()$ is linearly bounded, \emph{i.e.} $g(k) \leq ck$ for some
% constant $c$, then we say that the reduction is a \emph{linear
%   parameterized reduction}. Two basic classes of parameterized
% intractability are W[1] and~W[2]; if there is a parameterized
% reduction from a W[1]-hard  problem to a parameterized
% problem~$L$, then this implies W[1]-hardness of~$L$.
% % 
% % 
% %  We make use of the
% % following well-known result of Chen
% % \emph{et al.}~\cite{ChenCFHJKX05}. \todo{maybe define the multi-colored clique problem here already?}
% % 
% % \begin{lemma}
% % \label{lem:lower-bound}%
% % Let~$L$ be a parameterized problem with parameter $k$, and assume that
% % there is a linear parameterized reduction from \textsc{Multicolored
% %   Clique} to~$L$. Then, unless all problems in \textnormal{SNP} can be
% % solved in subexponential time, $L$ cannot be solved by an algorithm
% % running in $n^{o(k)}$ time.
% % \end{lemma}
% % The assumption that there are problems in \textnormal{SNP}
% % that cannot be solved in subexponential time is known as Exponential Time Hypothesis (ETH), refer to~\cite{LMS11b} for a recent survey on ETH-based running time lower bounds.

\section{Clique Cover Number and Domination Number}\label{sec:clique-cover}
In this section, we prove that on graphs of \emph{diameter} at most three, \tclub is NP-hard even if either the \emph{minimum clique cover number} is three or the \emph{domination number} is two. We first show that these bounds are tight.
The size of a maximum independent set is at most the size of a \emph{minimum clique cover}. 
Moreover, since each maximal independent set is a dominating set, the \emph{domination number} is also at most the size of a \emph{minimum clique cover}.

\begin{lemma}\label{lem:NP-hard-ind-two}
	For $s\ge 2$, \sclub is polynomial-time solvable on graphs where the size of a maximum independent set is at most two.
\end{lemma}
% \appendixproof{\autoref{lem:NP-hard-ind-two}}
{
% To prove \autoref{lem:NP-hard-ind-two}, we need the following observation.
% \begin{obs}\label{obs:ext-lemma}
%   Let~$v$ be a vertex such that~$N_s(v)$ is a $s$-club for an integer~$s>0$. Then, any maximal $(s+1)$-club that contains~$v$ also contains $N_s(v)$.
% \end{obs}
% \begin{proof}[of \autoref{lem:NP-hard-ind-two}]
% 	Let~$S$ be an $(s+1)$-club that contains~$v$ but not all vertices from~$N_s(v)$. We show that~$S\cup N_s[v]$ is also a $(s+1)$-club.
% 	Obviously, all vertices from~$N_s[v]$ are within distance~$s+1$ in~$G[S\cup N_s[v]]$. Similarly, all vertices from~$S$ are within~$s+1$ in~$G[S\cup N_s[v]]$. Thus, consider a vertex $w\in S\setminus N_s[v]$. By definition, a shortest path from~$w$ to~$v$ has length at most~$s+1$ and thus there is a vertex $u\in N(w)\cap N_s(v)$. Since $N_s(v)$ is an $s$-club, there is a path of length at most~$s$ from~$u$ to any vertex in~$N_s(v)$. Consequently,~$w$ has distance at most $s+1$ to any vertex in $N_s[v]$. Thus, $S\cup N_S[v]$ is an~$(s+1)$-club.
% \end{proof}

\begin{proof}  
   Let~$G=(V,E)$ be a graph. If a maximum independent set in~$G$ has size one or~$G$ has diameter~$s$, then~$V$ is an \mbox{$s$-club}. Otherwise %, if the maximum independent set in~$G$ is of size two, then 
   iterate over all possibilities to choose two vertices~$v,u\in V$. Denoting by~$G'$ the graph that results from deleting $N(v)\cap N[u]$ in~$G$, output a maximum size set $N^{G'}[v]\cup (N^{G'}(u)\cap N^{G'}_s(v))$ among all iterations.

  We next prove the correctness of the above algorithm. For a maximum size $s$-club $S \subseteq V$ in~$G$, there are two vertices~$v,u\in V$ such that $v\in S$ and~$\dist_{G[S\cup \{u\}]}(v,u)>s$, implying that $N(v)\cap N[u]\cap S=\emptyset$. Moreover, $N^{G'}[v]$ and $N^{G'}[u]$ are cliques: Two non-adjacent vertices in~$N^{G'}(v)$ (in~$N^{G'}(u)$) would form together with~$u$ (with~$v$) an independent set. 
  Since $N^{G'}[v]$ is a clique and~$v\in S$, $G[S\cup N^{G'}(v)]$ is a $s$-club and thus $N^{G'}[v]\subseteq S$ by the maximality of~$S$. Moreover, since $\{v,u\}$ is a maximum independent set and thus also a dominating set it remains to specify $N^{G'}(u)\cap S$.   
  However, since $N^{G'}[u]$ and $N^{G'}[v]$ are cliques and each vertex in $N^{G'}(u)\cap S$ needs to have distance at most~$s$ to~$v$, each vertex from $N^{G'}(u)\cap N_s(v)$ is contained in~$S$, implying that $S=N^{G'}[v]\cup (N^{G'}(u)\cap N^{G'}_s(v))$.
\end{proof}%
% }
% 
The following theorem shows that % for polynomial time solvability
the bound on the maximum independent set size in~\autoref{lem:NP-hard-ind-two} is tight.
% edges. 
\begin{theorem}\label{thm:NP-hard-independent-set}
	\tclub is NP-hard on graphs with \emph{clique cover number} three and \emph{diameter} three.
\end{theorem}
% \appendixproof{\autoref{thm:NP-hard-independent-set}}{
\begin{proof}
	We describe a reduction from \Clique.
	Let~$(G=(V,E),k)$ be a \Clique instance.
	We construct a graph~$G'=(V',E')$ consisting of three disjoint vertex sets, that is,~$V' = V_1 \cup V_2 \cup V_E$.
	Further, for $i \in \{1,2\}$, let~$V_i = V^V_i \cup V^{\rm{big}}_i$, where~$V^V_i$ is a copy of~$V$ and~$V^{\rm{big}}_i$ is a set of~$n^5$ vertices.
	Let~$u,v \in V$ be two adjacent vertices in~$G$ and let~$u_1,v_1 \in V_1$, $u_2,v_2 \in V_2$ be the copies of~$u$ and~$v$ in~$G'$.
	Then add the vertices~$e_{uv}$ and~$e_{vu}$ to~$V_E$ and add the edges~$\{v_1, e_{vu}\}, \{e_{vu}, u_2\}, \{u_1, e_{uv}\}, \{e_{uv}, v_2\}$ to~$G'$.
	Furthermore, add for each vertex~$v \in V$ the vertex set~$V^v_E = \{e_v^1, e_v^2, \ldots, e_v^{n^3}\}$ to~$V_E$ and make~$v_1$ and~$v_2$ adjacent to all these new vertices.
	Finally, make the following vertex sets to cliques:~$V_1$, $V_2$, $V_E$, and~$V_1^{\rm{big}} \cup V_2^{\rm{big}}$. Observe that~$G'$ has diameter three and that it has a \emph{clique cover number} of three.

	We now prove that~$G$ has a clique of size~$k$ $\Leftrightarrow$~$G'$ has a 2-club of size~$k' = 2n^5+kn^3+2k+2\binom{k}{2}$.

	``$\Rightarrow$:'' Let~$S$ be a clique of size~$k$ in~$G$.
	Let~$S_c$ contain all the copies of the vertices of~$S$.
	Furthermore, let~$S_E := \{e_{uv} \mid u_1 \in S_c \wedge v_2 \in S_c\}$ and $S_b := \{e_v^i \mid v \in S \wedge 1 \leq i \leq n^3\}$.
	We now show that~$S' := S_c \cup S_E \cup S_b \cup V_1^{\rm{big}} \cup V_2^{\rm{big}}$ is a 2-club of size~$k'$.
	First, observe that~$|V_1^{\rm{big}} \cup V_2^{\rm{big}}| = 2n^5$ and~$|S_c| = 2k$.
	Hence, $|S_b| = kn^3$ and~$|S_E| = 2\binom{k}{2}$.
	Thus, $S'$ has the desired size.
% 	Obviously,~$S'$ is a 2-club. \todo{genauer erklären}
	With a straightforward case distinction one can check that~$S'$ is indeed a 2-club.

	``$\Leftarrow$:'' Let~$S'$ be a 2-club of size~$k'$.
	Observe that~$G'$ consists of~$|V'| = 2n^5 + 2n + 2m + n^4$ vertices.
	Since~$k' > 2n^5$ at least one vertex of~$V_1^{\rm{big}}$ and of~$V_2^{\rm{big}}$ is in~$S$.
	Since all vertices in~$V_1^{\rm{big}}$ and in~$V_2^{\rm{big}}$ are twins, we can assume by \autoref{obs:twin-containment} that all vertices of~$V_1^{\rm{big}} \cup V_2^{\rm{big}}$ are contained in~$S'$.
	Analogously, it follows that at least~$k$ sets~$V^{v^1}_E, V^{v^2}_E, V^{v^3}_E, \ldots, V^{v^k}_E$ are completely contained in~$S'$.
	Since~$S'$ is a 2-club, the distance from vertices in~$V_i^{\rm{big}}$ to vertices in~$V^{v^j}_E$ is at most two.
	Hence, for each set~$V^{v^j}_E$ in~$S'$ the two neighbors~$v_1^j$ and~$v_2^j$ of vertices in~$V^{v^j}_E$ are also contained in~$S'$.
	Since the distance of~$v_1^i$ and~$v_2^j$ for~$v_1^i,v_2^j \in S'$ is also at most two, the vertices~$e_{v^iv^j}$ and~$e_{v^jv^i}$ are part of~$S'$ as well. Consequently,~$v^i$ and~$v^j$ are adjacent in~$G$.
	Therefore, the vertices~$v^1, \ldots, v^k$ form a size-$k$ clique in~$G$.
\end{proof}
% }
Since a maximum independent set is also a dominating set, 
\autoref{thm:NP-hard-independent-set} implies that \tclub is NP-hard on graphs with \emph{domination number} three and \emph{diameter} three. In contrast, for \emph{domination number} one \tclub is trivial. The following theorem shows that this cannot be extended.
% \tclub is NP-hard for dominating set of size two and diameter three.
\begin{theorem}\label{thm:NP-hard-dom}
 	\tclub is NP-hard even on graphs with \emph{domination number} two and \emph{diameter} three.
\end{theorem}
% \appendixproof{\autoref{thm:NP-hard-dom}}{
\begin{proof}
  % We extend the reduction from \Clique to 2-club which was given in
  % the proof of \autoref{thm:NP-hard-split-diameter}. Recall that,
  % given a \Clique instance $(G,k)$ the reduction forms a split graph
  % $G'=(I\cup C,E')$ with an independent set~$I$ and a clique~$C$
  % such that $(G',|C|+k)$ is a yes-instance of 2-club iff there is a
  % size-$k$ clique in~$G$. We extend the graph as follows.

  We present a reduction from \textsc{Clique}.  Let $(G=(V,E),k)$ be
  a~\textsc{Clique} instance and assume that~$G$ does not contain isolated vertices.  We construct the graph~$G'$ as
  follows. First copy all vertices of~$V$ into~$G'$. In~$G'$ the
  vertex set~$V$ will form an independent set.  Now, for each edge
  $\{u,v\}\in E$ add an \emph{edge-vertex} $e_{\{u,v\}}$ to~$G'$
  and make~$e_{\{u,v\}}$ adjacent to~$u$ and~$v$. Let~$V_E$ denote the
  set of edge-vertices. Next, add a vertex set~$C$ of size~$n+2$
  to~$G'$ and make~$C\cup V_E$ a clique.
  % The construction of~$G'$ is completed by making~$C'$ a clique.
  Finally, add a new vertex~$v^*$ to~$G'$ and make~$v^*$ adjacent to
  all vertices in~$V$.  Observe that~$v^*$ plus an arbitrary vertex
  from~$V_E\cup C$ are a dominating set of~$G'$ and that~$G'$ has
  diameter three.
  % Observe that $\{v\}\cup I$ now is a 2-club of size
  % $n+1$. % To prevent that this 2-club forms a feasible
  % solution, extend the clique~$C$ to the clique~$C'$ by adding some
  % new vertices such that $|C'|> n+1 + |C|$.
  We complete the proof by showing that~$G$ has a clique of size~$k$
  $\Leftrightarrow$~$G'$ has a 2-club of size at least $|C|+|V_E|+k$.

  ``$\Rightarrow$:'' Let~$K$ be a size-$k$ clique in~$G$.  Then,~$S:=K\cup C\cup V_E$ is a size-$|C|+|V_E|+k$ 2-club in~$G$:
  First, each vertex in~$C\cup V_E$ has distance two to all other
  vertices~$S$. Second, each pair of vertices~$u,v\in K$
  is adjacent in~$G$ and thus they have the common neighbor~$e_{\{u,v\}}$
  in~$V_E$.  

  ``$\Leftarrow$:'' Let~$S$ be a 2-club of size $|C|+|V_E|+k$ in~$G'$.
  Since $|C|> |V\cup \{v^*\}|$, it follows that there is at
  least one vertex $c\in S\cap C$. Since~$c$ and~$v^*$ have distance
  three, it follows that $v^*\not\in S$. Now since~$S$ is a 2-club,
  each pair of vertices $u,v\in S\cap V$ has at least one common
  neighbor in~$S$. Hence,~$V_E$ contains the
  edge-vertex~$e_{\{u,v\}}$. Consequently,~$S\cap V$ is a size-$k$
  clique in~$G$.
\end{proof}
% }

\section{Distance to Bipartite Graphs}\label{sec:dist-bipartite}
A 2-club in a bipartite graph is a biclique  and, thus, \tclub is polynomial-time solvable on bipartite graphs~\cite{Sch09}.
% Hence, the problem of finding a 2-club of size $\size$ in a bipartite graph~$G$ is equal to find a biclique (a complete bipartite subgraph) with~$\size$ vertices. 
However, \tclub is already NP-hard on graphs that become bipartite by deleting only one vertex.

\begin{theorem}\label{thm:np-hard-dist-bipartite}
  \tclub is NP-hard even on graphs with distance one to bipartite graphs.
\end{theorem}
% \appendixproof{\autoref{thm:np-hard-dist-bipartite}}{
\begin{proof}
 We reduce from the NP-hard \textsc{Maximum 2-SAT} problem: Given a positive integer~$k$ and a set $\C:=\{C_1, \ldots, C_m\}$ of clauses over a variable set~$X = \{x_1, \ldots, x_n\}$ where each clause~$C_i$ contains two literals, the question is whether there is an assignment~$\beta$ that satisfies at least~$k$ clauses.

% \decprob{\textsc{Maximum 2-SAT}}%
% {A set~$\C:=\{C_1, \ldots, C_m\}$ of clauses over a variable set~$X = \{x_1, \ldots, x_n\}$ where each clause~$C_i$ contains two literals.}  {Is there an assignment~$\beta$ for~$X$ that satisfies at least~$k$ clauses of~$\C$?}
%   \begin{quote}
%     Input: A set~$\C:=\{C_1, \ldots, C_m\}$ of clauses over a variable set~$X = \{x_1, \ldots, x_n\}$ where each clause~$C_i$ contains two literals.\\
%     Question: Is there an assignment~$\beta$ for~$X$ that satisfies at least~$k$ clauses of~$\C$? \\
%   \end{quote}
  Given an instance of~\textsc{Maximum 2-SAT} where we assume that each clause occurs only once, we construct an undirected graph~$G=(V,E)$.
The vertex set~$V$ consists of the four
  disjoint vertex sets~$V_\C$,~$V_F$,~$V_X^1$,~$V^2_X$, and
  one additional vertex~$v^*$. The construction of the four subsets
  of~$V$ is as follows.

  The vertex set~$V_\C$ contains one vertex~$c_i$ for each
  clause~$C_i\in \C$. The vertex set~$V_F$ contains for each
  variable~$x\in X$ exactly~$n^5$ vertices~$x^1\ldots x^{n^5}$. The
  vertex set~$V_X^1$ contains for each variable~$x\in X$ two
  vertices:~$x_t$ which corresponds to assigning true to~$x$ and~$x_f$
  which corresponds to assigning false to~$x$. The vertex set~$V_X^2$
  is constructed similarly, but for every variable~$x\in X$ it
  contains~$2\cdot n^3$ vertices:~the vertices~$x^1_t,\ldots
  x^{n^3}_t$ which correspond to assigning true to~$x$, and the
  vertices~$x^1_f,\ldots x^{n^3}_f$ which correspond to assigning
  false to~$x$.
  
  Next, we describe the construction of the edge set~$E$. The
  vertex~$v^*$ is made adjacent to all vertices in~$V_\C\cup
  V_F\cup V_X^1$. Each vertex~$c_i\in V_\C$ is made adjacent
  to the two vertices in~$V_X^1$ that correspond to the two literals
  in~$C_i$. Each vertex~$x^i\in V_F$ is made adjacent to~$x_t$
  and~$x_f$, that is, the two vertices of~$V_X^1$ that correspond to
  the two truth assignments for the variable~$x$. Finally, each
  vertex~$x^i_t\in~V_X^2$ is made adjacent to all vertices of~$V_X^1$
  except to the vertex~$x_f$. Similarly, each~$x^i_f\in~V_X^2$ is made
  adjacent to all vertices of~$V_X^1$ except to~$x_t$. This completes
  the construction of~$G$ which can clearly be performed in polynomial
  time. Observe that the removal of~$v^*$ makes~$G$ bipartite: each of
  the four vertex sets is an independent set and the vertices
  of~$V_\C$,~$V_F$, and~$V_X^2$ are only adjacent to vertices
  of~$V_X^1$.

  The main idea behind the construction is as follows. The size of the
  2-club forces the solution to contain the majority of the vertices in~$V_F$ and~$V_X^2$.
  % This has two consequences: For each~$x\in X$
  % either~$x_t$ or~$x_f$ must be contained in a 2-club; otherwise two
  % vertices from~$V_F$ and~$V_X^2$ have distance three. Furthermore,
  % since~$V_X^2$ has only neighbors in~$V_X^1$ and since the subgraph
  % induced by~$V_X^1\cup V_X^2$ is bipartite, the subsets of~$V_X^1$
  % and~$V_X^2$ in a 2-club induce a complete bipartite
  % graph. Accordingly,
  As a consequence, for each~$x\in X$ exactly one of~$x_t$ or~$x_f$ is
  in the 2-club. Hence, the vertices from~$V_X^2$ in the 2-club
  represent a truth assignment. In order to fulfill the bound on the
  2-club size, at least~$k$ vertices from~$V_\C$ are in the
  2-club; these vertices can only be added if the corresponding
  clauses are satisfied by the represented truth assignment. It remains to prove the following claim:
%   \begin{claim}\label{claim:correctness-bipartite}
% \vspace{3pt}

	\medskip

  \noindent \emph{Claim.} $(\C,k)$ is a yes-instance of~\textsc{Maximum 2-Sat}
    $\Leftrightarrow$~$G$ has a 2-club of size~$n^6+n^4+n+k+1$.
%   \end{claim}
% \appendixproof{Claim in the proof of \autoref{thm:np-hard-dist-bipartite}}

\medskip

\emph{Proof.}
  ``$\Rightarrow$'': Let~$\beta$ be an assignment for~$X$ that
  satisfies~$k$ clauses~$C_1, \ldots, C_k$ of~$\C$. Consider
  the vertex set~$S$ that consists of~$V_F$,~$v^*$, the vertex
  set~$\{c_1, \ldots, c_k\}\subseteq V_\C$ that corresponds
  to the~$k$ satisfied clauses, and for each~$x\in X$ of the vertex
  set~$\{x_t, x_t^1,\ldots, x_t^{n^3}\}\subseteq V^1_X\cup V^2_X$ if~$\beta(x)=\textrm{true}$
  and the vertex set~$\{x_f, x_f^1,\ldots, x_f^{n^3}\}\in V^1_X\cup V^2_X$
  if~$\beta(x)=\textrm{false}$. Clearly,~$|S|=n^6+n^4+n+k+1$. In the
  following, we show that~$S$ is a 2-club. Herein,
  let~$S_X^1:=V_X^1\cap S$,~$S_X^2:=V_X^2\cap S$,
  and~$S_\C:=V_\C\cap S$.

  First,~$v^*$ is adjacent to all vertices in~$S_\C\cup V_F
  \cup S_X^1$. Hence,~all vertices of~$S\setminus S_X^2$ are within
  distance two in~$G[S]$. By construction, the vertex sets~$S_X^1$
  and~$S_X^2$ form a complete bipartite graph in~$G$: A
  vertex~$x_t^i\in S_X^2$ is adjacent to all vertices in~$V_X^1$
  except~$x_f$ which is not contained in~$S_X^1$. The same argument
  applies to some~$x_f^i\in S_X^2$. Hence, the vertices of~$S_X^2$ are
  neighbors of all vertices in~$S_X^1$. This also implies that the
  vertices of~$S_X^2$ are in~$G[S]$ within distance two from~$v^*$ and
  from every vertex in~$V_F$ since each vertex of~$V_F\cup~\{v^*\}$ has
  at least one neighbor in~$S_X^1$. Finally, since the~$k$ vertices
  in~$S_\C$ correspond to clauses that are satisfied by the
  truth assignment~$\beta$, each of these vertices has at least one
  neighbor in~$S_X^1$. Hence, every vertex in~$S_X^2$ has
  in~$G[S]$~distance at most two to every vertex in~$S_\C$.

  ``$\Leftarrow$'': Let~$S$ be a 2-club of size~$n^6+n^4+n+k+1$, and
  let~$S_X^1:=V_X^1\cap S$,~$S_X^2:=V_X^2\cap S$,~$S_F:=V_F \cap S$
  and~$S_\C:=V_\C\cap
  S$. Clearly, neither $S_X^2=\emptyset$ nor $S_F=\emptyset$.

  Since~$|V_\C|+|V_X^1|+|V_X^2|+1\le n^2+2n+2n^4+1< n^5$
  for sufficiently large~$n$, $S$~contains more than~$n^6-n^5$
  vertices from~$V_F$. Consequently, for each~$x\in X$ there is an index $1\le i\le n^5$ such that $x^i\in S_F$.

  We next show that for each $x\in X$ it holds that either~$x_t$ or~$x_f$ is contained in~$S_X^1$. Towards this, since~$S$ is a 2-club, every vertex pair~$x^i\in S_F$ and~$u\in S_X^2$ has at least one common neighbor in~$S$. By construction, this common neighbor is a vertex of~$S_X^1$ and thus either~$x_t$ or~$x_f$. Moreover, by the observation above for each $x\in X$ at least one $x^i$ is contained in~$S_F$. Thus, for each $x\in X$ at least one of~$x_t$ and~$x_f$ is contained in~$S_X^1$. 

  Now observe that,~$G[S_X^1\cup S_X^2]$ is a complete bipartite
  graph, since~$S_X^1$ and~$S_X^2$ are independent sets and~$S_X^2$
  has only neighbors in~$S_X^1$. This implies that if for some~$x\in X$ there exists indices $1\le i,j\le n^3$ with $x_t^i$ and~$x_f^j$ are in~$S_X^2$, then~$x_t$ and~$x_f$ are \emph{not} in~$S_X^1$. This contradicts the above observation that at least one of~$x_t$ and~$x_f$ is in~$S_X^1$. 
  Moreover,  since~$|V_\C|+|V_X^1|+1\le n^2+2n+1<n^3$ and $|S\setminus V_F|> n^4$, we have~$|S_X^2|> n^4-n^3$. It follows that for each $x\in X$ there is an index $1\le i\le n^3$ such that either $x^i_t\in S_X^2$ or $x^i_f\in S_X^2$. Finally, this implies that either~$x_t$ or~$x_f$ is not contained in~$S_X^1$. 

  Summarizing,~$S$ has at most~$n^6$ vertices from~$V_F$, at most~$n^4$ vertices belonging to~$S_X^2$, exactly~$n$ vertices belonging to~$S_X^1$, and thus there are~$k+1$ vertices in~$S_\C\cup \{v^*\}$.   Since~$S$ is a~2-club that has nonempty~$S_X^2$, every one of the at least~$k$ vertices from~$S_\C$ has at least one neighbor in~$S_X^1$. Because for each~$x\in X$ either~$x_f$ or~$x_t$ is in~$S_X^1$, the~$n$ vertices from~$S_X^1$ correspond to an assignment~$\beta$ of~$X$.  By the above observation, this assignment satisfies at least~$k$ clauses of~$\C$.
\end{proof}
% }

\section{Average Degree and $h$-Index}\label{sec:h-index}

\newcommand{\XPhindexAlgRunTime}{O(2^{k^4}\cdot n^{2^k}\cdot n^2m)}
\tclub is fixed-parameter tractable for the parameter \emph{maximum degree} (the algorithm of \citet{SKMN11} can be analyzed in that way without any changes).  
It has been observed that in large-scale biological~\cite{JTAOB00} and social networks~\cite{BA99} the degree distribution  often follows a power law, implying that there are some high-degree vertices while most vertices have low degree.
This suggests considering stronger, that is, provably smaller, 
parameters such as \emph{$h$-index}, \emph{degeneracy}, and \emph{average degree}. 
For any graph it holds that $\text{avg. degree}\le 2 \cdot \text{degeneracy}\le 2\cdot h$-index, see also \autoref{fig:result-overview} for other relationships.
Furthermore, analyzing the coauthor network derived from the DBLP dataset\footnote{The dataset and a corresponding documentation are available online (\url{http://dblp.uni-trier.de/xml/}). Accessed Feb.~2012} with more than 715,000 vertices, maximum degree~804, $h$-index~208, degeneracy~113, and average degree~7 shows that also in real-world social networks these parameters are considerably smaller than the maximum degree (see~\cite{HKN12} for an analysis of these parameters on a broader dataset).

Unsurprisingly,  \tclub is NP-hard even with constant \emph{average degree}.

\begin{proposition}\label{prop:NP-hard-constant-avg-degree}
 For any constant $\alpha>2$, \tclub is NP-hard on connected graphs with \emph{average degree} at most~$\alpha$.
\end{proposition}
\begin{proof}
  Let $(G,\size)$ be an instance of \tclub where~$\Delta$ is the
  maximum degree of~$G$. We can assume that~$\ell>\Delta+2$ since, as
  shown for instance in the proof of~\autoref{thm:NP-hard-independent-set},~\tclub remains
  NP-hard in this case. We add a path~$P$ to~$G$ and an edge from an
  endpoint~$p$ of~$P$ to an arbitrary vertex~$v\in
  V$. Since~$\ell>\Delta+2$, any $2$-club of size at least~$\ell$
  contains at least one vertex that is not in~$P$. Furthermore, it
  cannot contain~$p$ and~$v$ since in this case it is a subset of
  either~$N[v]$ or~$N[p]$ which both have size at most~$\Delta+2$ ($v$
  has degree at most~$\Delta$ in~$G$). Hence, the instances are
  equivalent.  Putting at least $\lceil\frac{2m}{\alpha-2}-n\rceil$
  vertices in~$P$ ensures that the resulting graph has average
  degree at most~$\alpha$.
\end{proof}

We remark that the bound provided in \autoref{prop:NP-hard-constant-avg-degree} is tight: Consider a connected graph~$G$ with average degree at most two, that is,~$\frac{1}{n} \sum_{v \in V} \deg(v) \le 2$. 
Since~$\sum_{v \in V} \deg(v) = 2m$, it follows that~$n \ge m$ and, thus, the feedback edge set of~$G$ contains at most one edge.
As \tclub is fixed-parameter tractable with respect to the (size of a) \emph{feedback edge set}~\cite{HKN12}, it follows that \tclub can be solved in polynomial time on connected graphs with \emph{average degree} at most two.

\autoref{prop:NP-hard-constant-avg-degree} suggests considering
``weaker'' parameters such as \emph{degeneracy} or
\emph{$h$-index}~\cite{ES09} of~$G$ (see \autoref{fig:result-overview}).  Recall that having $h$-index~$k$
means that there are at most~$k$ vertices with degree greater
than~$k$.  Since social networks have small $h$-index~\cite{HKN12}, fixed-parameter
tractability with respect to the~$h$-index would be
desirable. Unfortunately, we show that \tclub is W[1]-hard when
parameterized by the $h$-index and NP-hard with constant \emph{degeneracy}.
Following this result, we show that there is ``at
least'' an XP-algorithm implying that \tclub is  polynomial-time solvable for constant~$h$-index.

We reduce from the W[1]-hard \MCC problem~\cite{FHRV09}.

 \decprob{\MCC{}}%
{An undirected graph~$G = (V,E)$, $k\in \mathbb{N}$, and a (vertex) coloring $c:V\rightarrow\{1,\ldots,k\}$.}  {Is
  there a multicolored clique of size~$k$ in~$G$, that is, a clique~$C\subseteq V$ such that $c(v)\neq c(v')$ for all $\{v,v'\}\subseteq V$ with $v\neq v'$?}
 
\begin{lemma}\label{lem:h-index-hardness}
 There are two polynomial-time computable reductions that compute for any instance $(G,c,k)$ of  \MCC an equivalent \tclub-instance~$(G',\size)$ such that~$G'$ has \emph{diameter} three and, additionally, in reduction i)~$G'$ has \emph{$h$-index} at most~$k+7$ and in reduction~ii)~$G'$ has \emph{degeneracy} five.
\end{lemma}
\begin{proof}
  \crefname{enumi}{property}{properties}
  \Crefname{enumi}{Property}{Properties}
  The only difference between both reductions is the construction of a so-called \emph{coloring gadget}. We first describe the common part.
 Let $(G,c,k)$ with $G=(V,E)$ and $c:V\rightarrow\{1,\ldots,k\}$ be an instance of \MCC. We construct a graph~$G'$ and choose $\size\in \mathbb{N}$ such that $(G',\size)$ is a yes-instance for \tclub if and only if $(G,c,k)$ is a yes-instance for \MCC. 
%  Throughout the construction we will insert sets of vertices that will be twins in~$G'$, so-called \emph{twin sets}, and thus inserting an edge from a twin set to a vertex (or maybe even to another twin set) means to insert all edges between the corresponding vertices.
 We will first construct some structures in~$G'$ which allow to describe the basic ideas: 
 For each vertex~$v \in V$ create a \emph{vertex gadget} by adding the \emph{$\alpha$-vertices} $\{\alpha^v_1, \ldots, \alpha^v_n\}$, the \emph{$\beta$-vertices} $\{\beta^v_1, \ldots, \beta^v_{n+1}\}$, and the \emph{$\gamma$-vertices} $\{\gamma^v_1, \ldots, \gamma^v_n\}$, and $\{\omega^v_{\alpha},\omega^v_{\gamma}\}$. Add edges such that~$(\alpha^v_1,\beta^v_1,\gamma^v_1,\alpha^v_2,\beta^v_2,\gamma^v_2,\ldots,\alpha^v_n,\beta^v_n,\gamma^v_n,\omega^v_{\alpha},\beta^v_{n+1},\omega^v_{\gamma},\alpha^v_1)$ induces a cycle. Add the three vertices 
 $$U=\{u_\alpha,u_\beta,u_\gamma\}$$ 
 and add edges from all $\alpha$- ($\beta$-,$\gamma$-)vertices to~$u_\alpha$ ($u_\beta,u_\gamma$), respectively. 
 Add the edges $\{\omega^v_{\alpha}, u_\alpha\}$ and $\{\omega^v_{\gamma},u_\gamma\}$.
 Furthermore, for a fixed ordering $V=\{v_1, \ldots, v_n\}$ add for each edge $\{v_i, v_j\} \in E$ an \emph{edge-vertex}~$e_{v_i,v_j}$ that is adjacent to each of~$\{\alpha^{v_i}_{j},\beta^{v_i}_{j},\gamma^{v_j}_{i}\}$.
 (Observe that the $\alpha$- and $\gamma$-vertex neighbor are in different vertex gadgets.) 
 The following property is fulfilled:
 
 \begin{enumerate}
  \item For each vertex~$v$ in a vertex gadget it holds that $|N(v)\cap U|=1$ and $|N(v)\cap N(u)|=1$ for each $u\in U\setminus N(u)$.\label{en:VG-vertices-anchors}
   \end{enumerate}
 
 The idea of the construction is that~$U$ will be forced to be contained in any 2-club~$S$ of size at least~$\size$. Hence by \Cref{en:VG-vertices-anchors} it follows that if an $\alpha$-vertex in a vertex gadget is contained in~$S$, then the unique $\beta$- and $\gamma$-vertex in its neighborhood has to be contained in~$S$ as well. Since this argument symmetrically holds for~$\beta$- and $\gamma$-vertices, it follows that either all or none of the vertices from a vertex gadget are contained in~$S$. Observe that, in this context, $\omega^v_{\alpha}$ ($\omega^v_{\gamma}$) behaves like a ``normal'' $\alpha$- \mbox{($\gamma$-)} vertex. Analogously, each edge-vertex~$e_{v_i,v_j}\in S$ needs to have a common neighbor with each vertex of~$U$. Thus, $e_{v_i,v_j}\in S$ implies that all vertices in the two vertex gadgets that correspond to~$v_i$ and~$v_j$ are contained in~$S$. 
 By connecting the vertices~$\{\omega^v_{\alpha},\omega^v_{\gamma}\}$ appropriately we will ensure that for each color~$c$ at most one vertex gadget whose vertex in~$G$ is colored with~$c$ can have a non-empty intersection with~$S$. (The construction of the corresponding coloring gadget is the only part where the two reductions differ.) Furthermore, we choose the value of~$\size$ such that~$S$ contains vertices from at least~$k$ vertex gadgets and at least~$\binom{k}{2}$ edge-vertices. Hence, there are exactly~$k$ vertex gadgets together with $\binom{k}{2}$ edge-vertices   that contribute to~$S$. Since the vertices corresponding to the vertex gadgets have different colors and since the endpoints of the edges corresponding to the~$\binom{k}{2}$ edge-vertices are all within this set of~$k$ vertices, the set~$S$ corresponds to a multicolored clique in~$G$. 
%  \pagebreak 
 
 To complete the construction and to ensure the properties discussed above, we next add the \emph{anchor gadget} and the \emph{coloring gadget}. To argue about their correctness we claim that, eventually, 
 \begin{equation}
      |V'| = \underbrace{n(3n+3)}_{n\text{ vertex gadgets}}+\underbrace{4n^3+7}_{\text{anchor gadget}}+\underbrace{m}_{m \text{ edge-vertices}}+\underbrace{|V_C|}_{\text{coloring gadget}} \label{eq:size-V'}
 \end{equation}
and we set \begin{equation}
\size:=\underbrace{k(3n+3)}_{k\text{ vertex gadgets}}+\underbrace{4n^3+7}_{\text{anchor gadget}} +\underbrace{\binom{k}{2}}_{\binom{k}{2} \text{ edge-vertices}}+\underbrace{|V_C|}_{\text{coloring gadget}}.          \label{eq:size-l}  
           \end{equation}
 
\noindent\textbf{Anchor Gadget:} We denote by~$V_A$ the set of all vertices in the anchor gadget including~$U$ and it will have size $4n^3+7$. Besides~$U$ the anchor gadget will contain only four other vertices, namely $\{l_U,l,r_1,r_2\}$, that have neighbors outside the gadget. Before describing the construction we will list some properties of it that will be used in the argumentation later on.
 \begin{enumerate}[resume]
  \item The set~$U$ is contained in any 2-club in~$G'$ of size at least~$\size$.\label{en:anchors}
  \item A 2-club of size at least~$\size$ contains either all or none of the vertices of a vertex gadget.\label{en:fully-vertex-gadget}
  \item For any two vertices $u\in U$ and $v\in \{r_1,r_2\}$ it holds that $(N(u)\cup N(v))\cap V_A=V_A\setminus (U \cup \{v\})$, that $(N(l)\cup N(v))\cap V_A=V_A\setminus U$, and $(N(l_U)\cup N(v))\cap V_A=V_A$. \label{en:s-neigh}
%   \item $N(u_c)\cup N(s')$ contains all vertices of the anchor gadget.\label{en:s-neigh}
 \end{enumerate}

 Informally, \Cref{en:s-neigh} ensures that if a vertex is adjacent to one of $\{l_U,l\}\cup U$ and to one of $\{r_1,r_2\}$, then it has distance at most two to all vertices in $V_A\setminus U$ and also distance at most two to all of~$U$ if it is adjacent to~$l_U$.
 
 	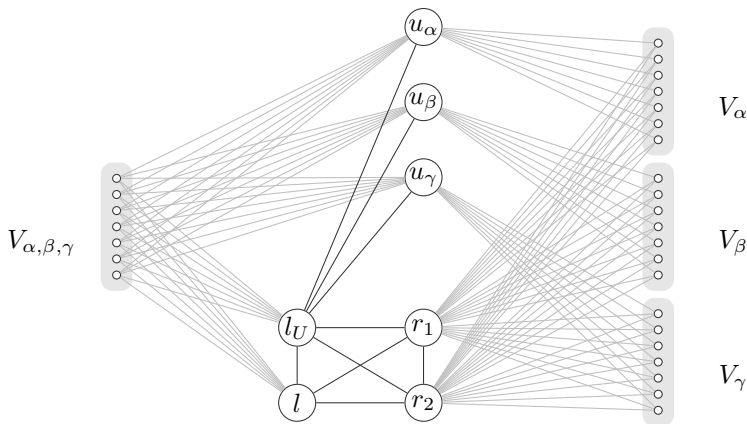
\begin{figure}[t]
		\begin{center}
		\def\layersep{2.4cm}
		\def\ylayersep{1cm}
		\def\numberOfSetVertices{7}
		\begin{tikzpicture}[draw=black!80, scale=1]
			\tikzstyle{vertex}=[circle,draw=black!80,minimum size=14pt,inner sep=0pt]
			\tikzstyle{thinedges}=[draw=black!25]
			\tikzstyle{small-vertex}=[circle,draw=black!80,minimum size=3pt,inner sep=0pt,fill=white]

			\node[vertex] (ua) at (0.7*\layersep,0) {$u_\alpha$};
			\node[vertex] (ub) at (0.7*\layersep,-\ylayersep) {$u_\beta$};
			\node[vertex] (ug) at (0.7*\layersep,-2*\ylayersep) {$u_\gamma$};

			\node[vertex] (s-) at 	(0.7*\layersep,-4*\ylayersep) {$r_1$};
			\node[vertex] (s) at 	(0.7*\layersep,-5*\ylayersep) {$r_2$};

			\node[vertex] (uc) at 	(0,-5*\ylayersep) {$l$};
			\node[vertex] (r2) at 	(0,-4*\ylayersep) {$l_U$};

			\foreach \y in {1,...,\numberOfSetVertices}{
				\node[small-vertex] (Vabc-\y) at (- \layersep, - \y * 1.5 * \ylayersep / \numberOfSetVertices - 1.8*\ylayersep) {};
				\path[thinedges] (ua) edge (Vabc-\y);
				\path[thinedges] (ub) edge (Vabc-\y);
				\path[thinedges] (ug) edge (Vabc-\y);
				\path[thinedges] (uc) edge (Vabc-\y);
				\path[thinedges] (r2) edge (Vabc-\y);
			}
			\node  [left of=Vabc-5]  {$V_{\alpha,\beta,\gamma}$};

			\foreach \y in {1,...,\numberOfSetVertices}{
				\node[small-vertex] (Va-\y) at (2*\layersep,- \y * 1.5 * \ylayersep / \numberOfSetVertices) {};
				\path[thinedges] (s) edge (Va-\y);
				\path[thinedges] (s-) edge (Va-\y);
				\path[thinedges] (ua) edge (Va-\y);
			}
			\node  [right of=Va-5]  {$V_{\alpha}$};
			\foreach \y in {1,...,\numberOfSetVertices}{
				\node[small-vertex] (Vb-\y) at (2*\layersep,- \y * 1.5 * \ylayersep / \numberOfSetVertices - 1.8*\ylayersep) {};
				\path[thinedges] (s) edge (Vb-\y);
				\path[thinedges] (s-) edge (Vb-\y);
				\path[thinedges] (ub) edge (Vb-\y);
			}
			\node  [right of=Vb-5]  {$V_{\beta}$};
			\foreach \y in {1,...,\numberOfSetVertices}{
				\node[small-vertex] (Vc-\y) at (2*\layersep,- \y * 1.5 * \ylayersep / \numberOfSetVertices - 3.6*\ylayersep) {};
				\path[thinedges] (s) edge (Vc-\y);
				\path[thinedges] (s-) edge (Vc-\y);
				\path[thinedges] (ug) edge (Vc-\y);
			}
			\node  [right of=Vc-5]  {$V_{\gamma}$};
			
			\path (r2) edge (uc);
			\path (r2) edge (s);
			\path (r2) edge (s-);
			\path (s) edge (uc);			
			\path (s) edge (s-);
			\path (s-) edge (uc);
			\path (r2) edge (ua);
			\path (r2) edge (ub);
			\path (r2) edge (ug);

			\begin{pgfonlayer}{background}
				\filldraw [line width=3mm,join=round,black!10]
				(Vabc-1.north  -| Vabc-1.east)  rectangle (Vabc-\numberOfSetVertices.south  -| Vabc-\numberOfSetVertices.west)
				(Va-1.north  -| Va-1.east)  rectangle (Va-\numberOfSetVertices.south  -| Va-\numberOfSetVertices.west)
				(Vb-1.north  -| Vb-1.east)  rectangle (Vb-\numberOfSetVertices.south  -| Vb-\numberOfSetVertices.west)
				(Vc-1.north  -| Vc-1.east)  rectangle (Vc-\numberOfSetVertices.south  -| Vc-\numberOfSetVertices.west);
			\end{pgfonlayer}
		\end{tikzpicture}
		\end{center}
		\caption{The anchor gadget. The vertices $\{u_\alpha,u_\beta,u_\gamma,l_U,l,r_1,r_2\}$ are the only vertices which have neighbors outside the anchor gadget. All the vertices in the sets $V_{\alpha,\beta,\gamma},V_\alpha,V_\beta$, and~$V_\gamma$ are twins and $u_\alpha$ ($u_\beta,u_\gamma$) is the only common neighbor between~$V_{\alpha,\beta,\gamma}$ and $V_\alpha$ ($V_\beta,V_\gamma$, resp.).\label{fig:connection_gadget_I}}		
	\end{figure}
 
 The anchor gadget is constructed as follows (see \autoref{fig:connection_gadget_I}): Add four sets $V_\alpha,V_\beta,V_\gamma,V_{\alpha,\beta,\gamma}$ each of size~$n^3$ and add edges from each vertex in~$V_{\alpha,\beta,\gamma}$ to each in $U\cup \{l,l_U\}$.
  Additionally, add edges from each vertex in~$V_\alpha$ to each of $\{u_\alpha,r_1,r_2\}$, from each vertex in~$V_\beta$ to each of $\{u_\beta,r_1,r_2\}$, and from each vertex in~$V_\gamma$ to each of $\{u_\gamma,r_1,r_2\}$.
   Finally, add edges such that $\{l_U,l,r_1,r_2\}$ is a clique and an edge from~$l_U$ to each vertex in~$U$.
 
	By the construction above, \Cref{en:s-neigh} is fulfilled and the anchor gadget is a 2-club. Observe that $u_\alpha$ ($u_\beta,u_\gamma$) is the only common neighbor of any vertex in~$V_{\alpha,\beta,\gamma}$ and any vertex in~$V_\alpha$ ($V_\beta,V_\gamma$, resp.) and hence if at least one vertex from each set $V_\alpha,V_\beta,V_\gamma,V_{\alpha,\beta,\gamma}$ is contained in a 2-club, then also~$U$ is contained. 
	To prove \Cref{en:anchors}, let~$S \subseteq V'$ be a 2-club of size~$\size$ that is disjoint to at least one of $\{V_\alpha ,V_\beta ,V_\gamma ,V_{\alpha,\beta,\gamma}\}$.
	The number of vertices that are \emph{not} in~$S$ is at most~$|V'|-\size$, which is (see \Cref{eq:size-V',eq:size-l}):
	\begin{equation*}
	 |V'|-\size =(n-k)(3n+3)+m-\binom{k}{2}<n^3.
	\end{equation*}
	
 This implies a contradiction and proves \Cref{en:anchors}. As argued above \Cref{en:VG-vertices-anchors,en:anchors} imply the correctness of \Cref{en:fully-vertex-gadget}.
 
  Recall that so far only~$U$ has neighbors outside the anchor gadget, namely all $\alpha$- ($\beta-,\gamma$-)vertices are adjacent to $u_\alpha$ ($u_\beta,u_\gamma$, resp.) and $\omega^v_{\alpha}$ ($\omega^v_{\gamma}$) is adjacent to~$u_\alpha$ ($u_\gamma$). We describe via properties how to connect the anchor gadget to the vertex gadgets.
  
  \begin{enumerate}[resume]
   \item $\omega^v_{\alpha}$ is adjacent to~$r_1$ and $\omega^v_{\gamma}$ is adjacent to~$r_2$ for all $v\in V$.\label{en:vertex-gadget-to-anchor-1}
   \item All $\alpha-,\beta-$, and $\gamma$-vertices and all edge-vertices are adjacent to each of $\{r_1,r_2\}$. Additionally, each edge-vertex is adjacent to~$l$.\label{en:vertex-gadget-to-anchor-2}
  \end{enumerate}
  Observe that \Cref{en:vertex-gadget-to-anchor-2} does not violate the correctness of \Cref{en:fully-vertex-gadget} since the vertices~$r_1,r_2$ are not neighbors of any vertex in~$U$ (see \Cref{en:s-neigh}).  
  \Cref{en:vertex-gadget-to-anchor-1,en:vertex-gadget-to-anchor-2,en:s-neigh} together imply 
  that all vertex pairs in~$G'$ except $\{\omega^v_{\alpha},\omega^{v'}_{\gamma}\}$ with $v\neq v'$  have distance at most two. We next construct the so-called coloring gadget that guarantees that only those vertex pairs  $\{\omega^v_{\alpha},\omega^{v'}_{\gamma}\}$ have a common neighbor (and thus can be contained in any 2-club) for which $c(v)\neq c(v')$. We will give two different constructions of the coloring gadget where the first guarantees an $h$-index of at most~$k+7$ and the second guarantees degeneracy five. Denoting the set of vertices in the coloring gadget by~$V_C$ both constructions fulfill the following properties:
  \begin{enumerate}[resume]
   \item Each vertex in~$V_C$ is adjacent to each of $\{l_U,r_1,r_2\}$.\label{en:V_C-anchor-gadget}
   \item Any pair $\{\omega^v_{\alpha},\omega^{v'}_{\gamma}\}$,~$v\neq v'$, has a common neighbor in~$V_C$ if and only if $c(v)\neq c(v')$.\label{en:omega-no-conflict}
  \end{enumerate}
  The two properties above are sufficient to prove the correctness of both reductions. 

  \textbf{Coloring gadget i):} 	For each color~$i \in \{1,\ldots,k\}$ add a vertex~$c_i$ and let~$V_C=\{c_1, \ldots , c_k\}$ the vertex set containing these vertices. Add an edge between a vertex $\omega^v_{\alpha}$ and~$c_i$ if $c(v)=i$ and an edge from $\omega^v_{\gamma}$ to~$c_i$ if $c(v)\neq i$ (\Cref{en:omega-no-conflict}). Finally, add edges such that each vertex in~$V_C$ is adjacent to each vertex in~$\{l_U,r_1,r_2\}$ (\Cref{en:V_C-anchor-gadget}).
  
  Note that the $h$-index of~$G'$ is at most $|V_C|+|U|+|\{l_U,l,r_1,r_2\}|=k+7$, as the vertices in $V_C\cup U \cup \{l_U,l,r_1,r_2\}$ are the only ones that might have degree at least~$k+7$.
  
  \textbf{Coloring gadget ii):} For each pair $\{\omega^{v}_{\alpha},\omega^{v'}_{\gamma}\}$ with $c(v)\neq c(v')$ add a vertex~$c_{v,v'}$ that is adjacent to each of~$\{\omega^{v}_{\alpha},\omega^{v'}_{\gamma}\}$ (\Cref{en:omega-no-conflict}). Finally, denoting all these new vertices by~$V_C$ we add an edge from each vertex in~$V_C$ to each vertex in~$\{l_U,r_1,r_2\}$ (\Cref{en:V_C-anchor-gadget}).

  We next prove that~$G'$ has degeneracy five by giving an elimination ordering, that is, an order of how to delete vertices of degree at most five that results in an empty graph: In the anchor gadget each of the vertices in $V_\alpha,V_\beta,V_\gamma,V_{\alpha,\beta,\gamma}$ has maximum degree five and hence they can be deleted. Then, delete all vertices in~$V_C$, as each of them also has degree five. Delete all edge-vertices (they also have degree five). In the remaining graph each vertex in a vertex gadget (see \Cref{en:vertex-gadget-to-anchor-2}) is adjacent to its two neighbors in its vertex gadget, adjacent to one of~$U$, and one or two neighbors in $\{r_1,r_2\}$. Hence, all vertices in vertex gadgets can be removed as they have degree at most five. The remaining vertices are $U\cup \{l_U,l,r_1,r_2\}$ and all vertices in $U\cup \{l,r_1,r_2\}$ have maximum degree four.  
  
  It remains to prove the correctness of the two reductions:\medskip

    \noindent \emph{Claim.} $(G,\c,k)$ is a yes-instance of
          \MCC~$\iff$~$(G',\size)$ is a yes-instance of
          \tclub.

   ``$\Rightarrow$`` Let~$C$ be a multicolored clique in~$G$ of size~$k$. We construct a set~$S\subseteq V'$ of size~$\size$ and prove that it is a 2-club in~$G'$. The set~$S$ contains each vertex gadget that corresponds to some vertex in~$C$, the coloring gadget, the anchor gadget, and any edge vertex~$e_{v_i,v_j}$ with $v_i,v_j\in C$. See \Cref{eq:size-l} to verify that $|S|=\size$. To verify that~$S$ is a 2-club, note that for each vertex~$v$ in a vertex gadget it holds that its unique neighbor with any vertex in $U\setminus N(v)$ is contained in~$S$ and thus from \Cref{en:s-neigh,en:vertex-gadget-to-anchor-1,en:vertex-gadget-to-anchor-2,en:VG-vertices-anchors} it follows that in~$G'[S]$ the vertex~$v$ has distance at most two to any anchor gadget vertex.
   Additionally, \Cref{en:V_C-anchor-gadget,en:vertex-gadget-to-anchor-1,en:vertex-gadget-to-anchor-2,en:omega-no-conflict} imply that~$v$ has distance at most two to all other vertex gadget vertices in~$S$, all coloring gadget vertices, and all edge vertices in~$S$. \Cref{en:vertex-gadget-to-anchor-2,en:s-neigh,en:V_C-anchor-gadget} imply that any coloring gadget vertex has distance at most two to all anchor vertices, coloring gadget vertices, and edge vertices. Finally, \Cref{en:vertex-gadget-to-anchor-2,en:s-neigh} show that each edge vertex has distance two to all anchor vertices.
 
 ''$\Leftarrow$`` Let~$S$ be a 2-club of size at least~$\size$. By \Cref{en:anchors} it follows that $U\subseteq S$ and by \Cref{en:fully-vertex-gadget} it follows that each vertex gadget is either fully contained in~$S$ or is disjoint to~$S$. Denote by~$C$ the vertices in~$G$ that correspond to the vertex gadgets that are fully contained in~$S$. First, since two vertices $\omega^v_{\alpha}$ and $\omega^{v'}_{\gamma}$,~$v\neq v'$, do not have a common neighbor if $c(v)=c(v')$ (\Cref{en:omega-no-conflict}) and there are only~$k$ colors, it follows that $|C|\le k$. Hence by \Cref{eq:size-V',eq:size-l} it follows that~$S$ contains at least~$\binom{k}{2}$ edge vertices. 
Since each edge vertex~$e_{v_i,v_j}$ needs to have a common neighbor with each vertex in~$U$ and the $\alpha$- and the $\gamma$- vertex neighbors of~$e_{v_i,v_j}$ are in different vertex gadgets, it follows that $\{v_i,v_j\}\subseteq C$. From this, since $|C|\le k$ it follows that $|C|=k$ and that~$S$ contains exactly $\binom{k}{2}$ edge vertices, implying that~$|C|$ induces a clique in~$G$. Finally, not that this clique is multicolored because of \Cref{en:omega-no-conflict}.
\end{proof}
\autoref{lem:h-index-hardness} imply several consequences.
\begin{corollary}
 \tclub is NP-hard on graphs with \emph{degeneracy} five.
\end{corollary}

\begin{corollary}\label{thm:w1-hard-h-index}
	\tclub parameterized by \emph{$h$-index} is W[1]-hard. 
\end{corollary}
Since the reduction in \autoref{lem:h-index-hardness} is 
from \MCC and in the reduction the new parameter is linearly bounded in the old one, the results of~\citet{CCF+05} imply the following.
\begin{corollary}
  \tclub cannot be solved in~$n^{o(k)}$-time on graphs with~\emph{$h$-index}~$k$ unless the exponential time hypothesis fails.
\end{corollary}

\newcommand{\T}{\mathcal{T}}
We next prove that there is an XP-algorithm for the parameter \emph{$h$-index}.
\begin{theorem}\label{thm:XP-h-index}
 \tclub can be solved in $\XPhindexAlgRunTime$ time where~$k$ is the~\emph{$h$-index} of the input graph.
\end{theorem}
% \appendixproof{\autoref{thm:XP-h-index}}{
\begin{proof}
 	\looseness=-1 We give an algorithm that finds a maximum 2-club in  a graph~$G'=(V',E')$ in $\XPhindexAlgRunTime$ time  where~$k$ denotes the $h$-index of~$G'$. 
	Let $X'\subseteq V'$ be the set of all vertices in~$G'$ with degree greater than~$k$. 
	By definition of the $h$-index, $|X'|\le k$. For the proof of correctness fix any maximum 2-club~$S$ in~$G'$.
	Throughout the algorithm via branching we will guess some vertices contained in~$S$ and we will collect them in the set~$P$. Then, cleaning the graph means to exhaustively remove all vertices that do not have distance at most two to all vertices in~$P$. These vertices cannot be contained in~$S$ and, clearly, if this requires to delete some vertex in~$P$ we will abort this branch.
	
	First, branch into the at most $2^k$ cases to guess the set $X=X'\cap S$ (potentially $X=\emptyset$). 
% 	In case~$X=\emptyset$, one can apply the fixed-parameter algorithm for the parameter maximum degree~\cite{SKMN11}.
	Delete all vertices from~$X'\setminus X$, initialize~$P$ with~$X$, and clean the graph. Denoting the resulting graph by~$G=(V,E)$, we next describe how to find a maximum 2-club in~$G$ that contains~$X$.
	Towards this, consider the at most~$2^k$ twin classes of the vertices in~$V\setminus X$ with respect to~$X$. Branch into the $O(n^{2^k})$ cases to guess for each twin class~$T$ any vertex from~$T\cap S$, called the \emph{center} of~$T$. Clearly, if $T\cap S=\emptyset$, then there is no center and we delete all vertices in~$T$. Add all the centers to~$P$ and clean the graph.
	
	Two twin classes~$T$ and~$T'$ are in \emph{conflict} if $N^{G}(T)\cap N^{G}(T')\cap X=\emptyset$. 
	Now, the crucial observation is that, if~$T$ and~$T'$ are in conflict, then all vertices in $(T\cup T')\cap S$ are contained in the same connected component of~$G[S\setminus X]$, since otherwise they would not have pairwise distance at most two. However, this implies that all vertices in $T\cap S$ have pairwise distance at most four in~$G[S\setminus X]$.
	Hence, for each twin class~$T$ with center~$c$ that is in conflict to any other twin class it holds that $T\cap S\subseteq N^{G-X}_4[c]$  and since $G-X$ has maximum degree at most~$k$, one can guess $N^S_4[c]:=N^{G-X}_4[c]\cap S$ by branching into at most $2^{k^4}$ cases. 
	Delete all vertices in~$T$ guessed to be not contained in~$N^S_4[c]$, add $N^S_4[c]$ to~$P$, and clean the graph.
	Note that the remaining graph is a 2-club, since~$P$ contains~$X$ and the intersection of~$S$ with each twin class that is in conflict to any other twin class. By definition of twin classes that are in conflict, it holds that all other twin classes share a common neighbor in~$X$.
\end{proof}
% }

% From \autoref{thm:w1-hard-h-index} it follows that \tclub is W[1]-hard with respect to degeneracy. It is open whether it is NP-hard with constant degeneracy. However, for the smaller parameter minimum degree this is indeed the case.
% 
% \begin{lemma}\label{prop:NP-hard-min-degree}
% 	\tclub is NP-hard even if~$G$ is connected and has minimum degree one.
% \end{lemma}
% % \appendixproof{\autoref{prop:NP-hard-min-degree}}{
% \begin{proof}
%  \tclub is NP-hard even if the minimum degree in the input graph is one:
% Given a graph~$G=(V,E)$ simply append a degree-one-vertex to an arbitrary vertex~$v\in V$ with $\deg(v) < \size - 2$.
% % Denote with~$G'$ the new graph.
% % The maximum 2-club containing the added degree-one-vertex~$u$ is the star with center~$v$. 
% % Since~$\deg(v) \le \size -2$ this 2-club is smaller than~$\size$ implying that~$G'$ contains a 2-club of size~$\size$ if and only if~$G'$ contains a 2-club of size~$\size$.
% The added vertex is not contained in any 2-club of size at
% least~$\size$ and, thus, the new graph contains a 2-club of size at
% least~$\size$ iff~$G$ contains one.
% \end{proof}fs
% % }

\section{\bf Distance to (Co-)Cluster Graphs and Cographs}\label{sec:fpt-dist-to-cocluster}
In this section we present fixed-parameter algorithms for \tclub
parameterized by \emph{distance to co-cluster graphs}, by \emph{distance to cluster
graphs}, and by \emph{distance to cographs}. All these algorithms have running
time $2^{O(2^k)}\cdot n^{O(1)}$ which is roughly similar to the one
obtained for treewidth~\cite{HKN12}. For the weaker parameters the
constants in the exponential part of the running time are
smaller. Hence, none of the algorithms ``dominates'' one of the other
algorithms even with \emph{distance to cographs} being a provably smaller
parameter than \emph{distance to cluster graphs} or \emph{distance to co-cluster}
graphs (see \autoref{fig:result-overview}).  As already mentioned,
even for the considerably weaker parameter \emph{vertex cover} the
best known algorithm has running time $2^{O(2^k)}\cdot
n^{O(1)}$. In contrast, the parameter \emph{distance to
clique} which is unrelated to \emph{vertex cover} admits a trivial $O(2^k\cdot
nm)$-time algorithm, even in case of the general \sclub. This is
implied by the $O(2^k\cdot nm)$-time algorithm for the dual parameter
$n-\size$~\cite{SKMN11} which can be interpreted as \emph{distance to
2-clubs} and the fact that each clique is a 2-club.
% 
% \begin{theorem}
%  \sclub is solvable in $O(2^k\cdot nm)$ time where~$k$ denotes the distance to a clique.
% \end{theorem}
% \begin{proof}
%  Let~$G$ be an input graph of an \sclub instance and let~$X$ be a vertex set of size~$k$ such that $G-X$ is a clique. The algorithm simply branches into the $2^k$ possibilities to guess the intersection of~$X\cap S$ with a maximum-size $s$-club~$S$ in~$G$. Afterwards, it deletes all vertices in $X\setminus (X\cap S)$ and exhaustively deletes all vertices that do not have distances at most~$s$ to all vertices in~$X\cap S$. If this requires to delete some vertex in~$X\cap S$, then abort this branch. Otherwise the remaining graph is a maximum-size $s$-club.
% \end{proof}

\subsection{Distance to Co-Cluster Graphs and Distance to Cluster Graphs}

To complete the picture drawn in \autoref{fig:result-overview}, we also give short description of an algorithm for \tclub parameterized by the \emph{distance to co-cluster graphs}. A graph is a co-cluster graph if its complement graphs is cluster graph.

\begin{theorem}\label{thm:dist-co-cluster}
	\tclub is solvable in~$O(2^k\cdot 2^{2^k}\cdot nm)$ time where~$k$ denotes the \emph{distance to co-cluster graphs}.
\end{theorem}
% \appendixproof{\autoref{thm:dist-co-cluster}}{
\begin{proof}
	Let $(G,X,\size)$ be an \tclub instance where~$X$ has $|X|=k$ and~$G-X$ is a co-cluster graph.
	Note that the co-cluster graph~$G-X$ is either a connected graph or an independent set.
	In the case that~$G-X$ is an independent set, the set~$X$ is a vertex cover
	and we thus apply the algorithm we gave in companion work~\cite{HKN12} to solve the instance in~$O(2^k\cdot 2^{2^k}\cdot nm)$ time.
	
	Hence, assume that~$G-X$ is connected.  Since~$G-X$ is the
	complement of a cluster graph, this implies that $G-X$ is a
	2-club. Thus, if~$\size \le n - k$, then we can trivially answer yes.
	Hence, assume that~$\size > n-k$ or, equivalently, $k > n -
	\size$.  \citet{SKMN11} showed that 2-club can be solved
	in~$O(2^{n-\size} nm)$  (simply choose a vertex
	pair having distance at least three and branch into the two
	cases of deleting one of them).  Since~$k > n-\size$ it
	follows that 2-club can be solved in~$O(2^k nm)$ time in this case. 
\end{proof}
% }

Next, we present a fixed-parameter algorithm for the parameter
\emph{distance to cluster graphs}.

\begin{theorem}\label{thm:fpt-dist-cluster}
  \tclub is solvable in~$O(2^k\cdot 3^{2^k}\cdot nm)$ time
  where~$k$ denotes \emph{distance to cluster graphs}.
\end{theorem}
% \appendixproof{\autoref{thm:fpt-dist-cluster}}
{\begin{proof}
  Let $(G,X,\size)$ be a \tclub instance where~$G-X$ is a
  cluster graph and ~$|X|=k$. First, branch into all possibilities to
  choose the subset $X'\subseteq X$ that is contained in the desired
  2-club~$S$. Then, remove $X\setminus X'$ and all vertices that are not
  within distance two to all vertices in~$X'$, and let~$G'=(V',E')$
  denote the resulting graph.

  Let ${\cal T}=T_1,\ldots,T_p$ be the set of twin classes of
  $V'\setminus X'$ with respect to~$X'$ and let~$C_1, \ldots, C_q$ denote the clusters of~$G'-X'$. Two twin
  classes~$T$ and~$T'$ are in \emph{conflict} if $N(T)\cap N(T')\cap
  X'=\emptyset$. 
  % Now, the crucial observation is that, if~$T$ and~$T'$ are in
  % conflict, then all vertices in $(T\cup T')\cap S$ are
  % contained in the same cluster of~$G'-X'$.
  % We define \emph{conflict classes} as the sets of twin classes that
  % are formed by the
  % transitive closure of the conflict relation between twin
  % classes. Twin classes that are not in conflict to any other twin
  % class are called \emph{non-conflicting twin classes} and they can
  % be
  % contained in multiple clusters of~$G'-X'$. Thus, it follows from
  % the
  % observation above that for a conflict class~$C$ all vertices in
  % $C\cap S$ are contained in the same connected component
  % of~$G'[S\setminus X']$.
  % Next, we partition the vertices in~$G-S^X$ into~$2^{|S^X}$
  % twin classes according to their neighborhoods in~$S^X$.
  % Furthermore, we try all possibilities of choosing the set of
  % twin classes that are contained in~$G-X$. Let~$T$ denote the
  % chosen set of twin classes. Observe that~$T$ contains two
  % types of twin classes: First, it contains vertices that have
  % with all other twin classes of~$T$, a common neighbor
  % in~$S^X$; we call these twin classes
  % \emph{non-conflicting}. Second, it contains twin classes~$C_i$
  % such that there is at least one further twin class~$C_j$ such
  % that the neighborhoods of the two twin classes in~$S^X$ have
  % empty intersection; we call these twin classes
  % \emph{conflicting} and say that~$C_i$ and~$C_j$ are in
  % \emph{conflict}.
  The three main observations exploited in the algorithm are the
  following: First, if two twin classes~$T_i$ and~$T_j$ are in
  conflict, then all vertices of~$T_i$ that are in a 2-club and all
  vertices from~$T_j$ that are in a 2-club must be in the same cluster
  of~$G'-X'$. Second, every vertex from~$G'-X'$ can reach all vertices
  in~$X'$ only via vertices of~$X'$ or via vertices in its own
  cluster. Third, if one 2-club-vertex~$v \in S$ is in a twin class~$v \in T_i$ and
  in a cluster~$v \in C_j$, then all vertices that are in~$T_i$
  and in~$C_j$ can be added to~$S$ without violating the 2-club property.

  We exploit these observations in a dynamic programming
  algorithm. In this algorithm, we create a two-dimensional
  table~$\cal{A}$ where an entry~${\cal A} [ i,{\cal T'}]$ stores the
  maximum size of a set~$Y\subseteq \bigcup_{1\le j\le i} C_j$ such
  that the twin classes of~$Y$ are \emph{exactly}~${\cal T'}\subseteq {\cal T}$
  and all vertices in~$Y$ have in~$G[Y\cup X']$ distance at most two
  to each vertex from~$Y\cup X'$.

  Before filling the table~${\cal A}$, we calculate a value~$s(i,{\cal
    T'})$ that stores the maximum number of vertices we can add from~$C_i$ that are from the twin classes in~$\cal{T'}$ and fulfill the requirements in the previous paragraph. This value is defined as follows. Let~$C_i^{\cal{T'}}$
  denote the maximal subset of vertices from~$C_i$ whose twin classes
  are exactly~${\cal T'}$. Then,~$s(i,{\cal T'})=|C_i^{\cal{T'}}|$
  if~$C_i^{\cal{T'}}$ exists and every pair of non-adjacent vertices
  from~$C_i^{\cal{T'}}$ and from~$X'$ have a common
  neighbor. Otherwise, set~$s(i,{\cal T'})=- \infty$. Note that as a
  special case we set~$s(i,\emptyset)=0$. Furthermore, for two
  subsets~${\cal T''}$ and~${\cal \tilde{T}}$ define the
  predicate~$\conf({\cal T''},{\cal \tilde{T}})$ as true if there is a
  pair of twin classes~$T_i\in {\cal T''}$ and~$T_j\in {\cal
    \tilde{T}}$ such that~$T_i$ and~$T_j$ are in conflict, and as
  false, otherwise.

  Using these values, we now fill~$\cal{A}$ with the following
  recurrence:
  \begin{multline*}
    {\cal A}[ i ,{\cal T'}]= \\ 
\mathop{\max}_{\cal{T''}\subseteq
      {T'},\cal{\tilde{T}}\subseteq {T'}}
    \begin{cases}
      {\cal A}[i-1,{\cal \tilde{T}}] + s(i,\cal{T''}) & \text{if } {\cal \tilde{T}} \cup {\cal T''}={\cal T'} \wedge \neg \conf({\cal \tilde{T}} ,{\cal T''}),\\
      -\infty & \text{otherwise.}
    \end{cases}
  \end{multline*}
  This recurrence considers all cases of combining a set~$Y$ for
  the clusters~$C_1$ to~$C_{i-1}$  with a solution~$Y'$ for the
  cluster~$C_i$. Herein, a positive table entry is only obtained when
  the twin classes of~$Y \cup Y'$ is exactly~${\cal T'}$ and the
  pairwise distances between~$Y\cup Y'$ and~$Y\cup Y'\cup X'$
  in~$G[Y\cup Y'\cup X']$ are at most two. The latter property is
  ensured by the definition of the~$s()$ values and by the fact that
  we consider only combinations that do not put conflicting twin classes
  in different clusters.
  
  Now, the table entry~${\cal A}[q,{\cal T'}]$ contains the size of a
  maximum vertex set~$Y$ such that in~$G'[Y\cup X']$ every vertex
  from~$Y$ has distance two to all other vertices. It remains to
  ensure that the vertices from~$X'$ are within distance two from each
  other. This can be done by only considering a table entry~${\cal A}[
  q, {\cal T'}]$ if each non-adjacent vertex pair~$x,x'\in X'$ has either a common
  neighbor in~$X'$ or in one twin class contained in~${\cal T'}$.  The
  maximum size of a 2-club in~$G'$ is then the maximum value of all
  table entries that fulfill this condition.  

  The running time can be bounded roughly as~$O(2^k\cdot 3^{2^k}\cdot
  nm)$: We try all~$2^k$ partitions of~$X$ and for each of these
  partitions, we fill a dynamic programming table with~$2^{2^k}\cdot
  n$ entries. The number of overall table lookups and updates is~$O(3^{2^k}\cdot n)$ since there are~$3^{2^k}$ possibilities to partition~$\cal T$ into the three sets~$\cal T''$, $\cal \tilde{T}$, and~$\cal T\setminus \cal T'$. Since each~$C_i$ is a clique, the entry~$s(i,{\cal T'})$ is computable in $O(nm)$ time and the overall running time follows.
\end{proof}}

\subsection{Distance to Cographs}
We proceed by describing the fixed-parameter algorithm with respect to the parameter \emph{distance to cographs}. Recall that since cographs are exactly the $P_4$-free graphs, any connected component of a cograph is a 2-club. 

\newcommand{\Pa}{\varUpsilon}
\begin{theorem}\label{thm:cographs}
 \tclub is solvable in   $O(2^k\cdot 8^{2^k}\cdot n^4)$ time where~$k$ denotes the \emph{distance to cographs}.
\end{theorem}
\begin{proof}
 Let~$G'$ be the input graph of a \tclub instance. Moreover, let~$X'$ be a vertex subset of~$G'$ whose deletion results in a cograph. We next describe a fixed-parameter algorithm with respect to $k=|X'|$ that finds  a maximum-size 2-club in~$G$. For our correctness proof we fix any maximum 2-club~$S$ in~$G'$. First branch into the at most $2^k$ cases to guess $X=X'\cap S$. Delete all vertices in~$X'\setminus X$.
 %  , and, correspondingly branch into the at most~$2^{2^k}$ cases to select all the twin classes in $V\setminus X$ that have a non-empty intersection with~$S$. 
%  Second exhaustively delete all vertices in~$G$ that do not have distance at most two to all vertices in~$X$ or are in a twin class guessed to be not contained in~$S$. (Clearly, if this requires to delete one vertex in~$X$, then we abort.) 
Denoting by $G=(V,E)$ the remaining graph, observe that $G-X$ is a cograph.
 The remaining task is to find a maximum 2-club in~$G$ that contains~$X$.
%  To solve this we describe a dynamic programming algorithm. 
 
 Before proceeding to describe the algorithm 
% we first need to define some additional notation: Consider the twin classes of $V\setminus X$ with respect to~$X$. We say that two twin classes~$T$ and~$T'$ are in \emph{conflict} and we call both twin classes \emph{conflicting}, if $N(T)\cap N(T')\cap X=\emptyset$. 
%  As otherwise they would have distance more than two it follows that if~$T$ and~$T'$ are in conflict, then there is one connected component~$C$ in~$G[S\setminus X]$ with $(T \cup T')\cap S\subseteq C$.
%  We use 
we introduce the following characterization of cographs~\cite{BLS99}: A graph is a cograph if it can be constructed from single vertex graphs by a sequence of parallel and series compositions. 
 Given~$t$ vertex disjoint graphs~$G_i=(V_i,E_i)$, the \emph{series composition} is the graph
 $(\bigcup_{i=1}^{t} V_i,\bigcup_{i=1}^{t} E_i \cup \{\{v,u\}\mid v\in G_i \wedge u\in G_j\wedge 1\le i<j\le t\}$ and the \emph{parallel composition} is $(\bigcup_{i=1}^{t}V_i, \bigcup_{i=1}^{t} E_i)$. 
 The corresponding \emph{cotree} of a cograph~$G$ is the tree whose leaves correspond to the vertices in~$G$ and each inner node represents a series or parallel composition of its children up to a root which represents~$G$.

 We next describe a dynamic programming algorithm that proceeds in a bottom-up manner on the cotree of $G-X$ and finds a maximum 2-club in~$G$ that contains~$X$.
 We may assume that $t=2$ for all series and parallel compositions, as otherwise we can simply split up the corresponding nodes in the cotree. For each node~$P$ in the cotree let $V(P)\subseteq V\setminus X$ be the vertices corresponding to the leaves of the subtree rooted in~$P$. Furthermore, consider the (at most~$2^k$ many) twin classes of $V\setminus X$ with respect to~$X$ and for a subset of twin classes~$\T$ let $V(\T)=\bigcup_{T\in \T} T$ denote the union of all vertices in the twin classes of~$\T$.
 We compute a table~$\Gamma$ where for any subset of twin classes~$\T$ and any node~$P$ of the cotree the entry~$\Gamma(P,\T)$ is the size of a largest set~$L\subseteq V(P)\cap V(\T)$ that fulfills the following properties:
 \begin{enumerate}
  \item for all $T\in \T:T\cap L\neq \emptyset$ and \label{prop-twin}
   \item for all $v\in L\cup X$ and $u\in L$: $\dist_{G[L\cup X]}(u,v)\le 2$ \label{prop-inner-X}
 \end{enumerate}
The intention of the definition above is that the graph $G[L\cup X]$ is a ``2-club-like'' structure that contains a vertex from each twin class in~$\T$ (\Cref{prop-twin}) and for any pair of vertices, except those where both vertices are from~$X$, have distance at most two (\Cref{prop-inner-X}). Denoting the root of the cotree by~$r$ and by~$\T_s$ the set of all twin classes that have a non-empty intersection with~$S$, $\Gamma(r,\T_s)\ge |S\setminus X|$ as~$S\setminus X$ trivially fulfills all properties.  Reversely, for any subset of twin class~$\T$ that contains for each pair of vertices $\{u,v\}\in X$ with $\dist_{G[X]}(u,v)>2$ a twin class $T\in \T$ with $\{u,v\}\subseteq N(T)$, any set corresponding to~$\Gamma(r,\T)$ forms together with~$X$ a 2-club.
 
 We now describe the dynamic programming algorithm. Let~$P$ be a leave node of the cotree with $V(P)=\{x\}$ and let $\T$ be any subset of twin classes. 
The two sets $\{x\}$ and~$\emptyset$ are the only candidates for~$L$. Hence we set $\Gamma(P,\T):=1$ if~${x}$ fulfills both properties,~$\Gamma(P,\emptyset):=0$ ($\emptyset$ fulfills both properties), and~$\Gamma(P,\T)=-\infty$ otherwise.%, if neither~$\{x\}$ nor~$\emptyset$ fulfills both properties.
  
  Next we describe the dynamic programming algorithm for inner nodes of the cotree.
  Let~$P$ be any node of the cotree with children~$P_1,P_2$ and let $\T$ be any subset of twin classes. We construct a graph~$G^P$ by exhaustively deleting in $G[(V(P)\cap V(\T))\cup X]$ all vertices from $V(P)\cap V(\T)$ that have distance more than two to any vertex in~$X$. (Clearly, such a vertex has to be deleted because of \Cref{prop-inner-X}.)
  If the resulting graph~$G^P$ violates \Cref{prop-twin}, then there is no set corresponding to $\Gamma(P,\T)$ and thus we set the entry to be~$-\infty$. Additionally, if~$G^P$ fulfills all properties, then set $\Gamma(P,\T)=|V(G^P)|-|X|$.
  To handle the remaining case where~$G^P$ violates only \Cref{prop-inner-X} we make a case distinction on the node type of~$P$.
  
  \emph{Case~1: $P$ is a series node.}\\
  Let $\{u,v\}\subseteq V(G^P)\setminus X$ be a vertex pair with $\dist_{G^P}(u,v)>2$. Since a series composition introduces an edge between each vertex in~$V(P_1)$ and each vertex in~$V(P_2)$ and $V(G^P)\subseteq V(P)=V(P_1)\cup V(P_2)$, it follows that either $V(G^P)\cap V(P_1)=\emptyset$ or $V(G^P)\cap V(P_2)=\emptyset$. 
 This implies that $\Gamma(P,\T)=\max\{\Gamma(P_1,\T),\Gamma(P_2,\T)\}$.
 
  \emph{Case~2: $P$ is a parallel node.}\\
  Consider any set~$L$ that corresponds to $\Gamma(P,\T)$. 
  By the definition of a parallel node there is no edge between a vertex from $V(P_1)$ to a vertex in~$V(P_2)$. Consequently, any pair of vertices in~$L$ with one vertex in~$V(P_1)$ and the other in~$V(P_2)$ have a common neighbor in~$X$. 
  Correspondingly, we say that two twin classes are \emph{consistent} if they have at least one common neighbor in~$X$ and two sets of twin classes are consistent if any twin class of the first set is consistent with any twin class of the second set. 
  Denoting by $\T^S_1$ ($\T^S_2$) the set of twin classes with a non-empty intersection with $L\cap V(P_1)$ ($L\cap V(P_2)$), by the argumentation above it follows that $\T^S_1$ is consistent with~$\T^S_2$.
  Additionally, it is straightforward to verify that $L\cap  V(P_1)$ ($L\cap  V(P_2)$) fulfills all properties (except being a largest set) for the entry $\Gamma(\T^S_1,P_1)$ ($\Gamma(P_2,\T^S_2)$). 

  Reversely, for any two consistent sets of twin classes $\T_1,\T_2$ let~$L_1$ ($L_2$) be any vertex set that corresponds to $\Gamma(P_1,\T_1)$ ($\Gamma(P _2,\T_2)$).
	It holds that $L_1\cup L_2$ fulfills all properties for $\Gamma(P,\T_1\cup \T_2)$ and hence $\Gamma(P,\tilde{\T}\cup \T_2)\ge |L_1\cup L_2|$.
  Hence it is correct to set $\Gamma(\T,P)$ to be the largest value of $\Gamma(P_1,\T_1)+\Gamma(P_2,\T_2)$ where $\T_1,\T_2$ are consistent and $\T_1 \cup \T_2=\T$. This completes the description of the algorithm.

  The table $\Gamma$ contains $O(n\cdot 2^{2^k})$ entries as there are at most $2^k$ twin classes. Each entry can be computed in $O(n^3+4^{2^k})$ time. In total, together with the factor of~$2^k$ needed guess~$X$, the running time of the above algorithm is 
%   $O(2^k\cdot (n^4\cdot 2^{2^k}+n2^{3\cdot 2^k}))$
  $O(2^k\cdot 8^{2^k}\cdot n^4)$.
\end{proof}

\section{Conclusion}
We have resolved the complexity status of \tclub for most of the parameters in the complexity landscape shown in \autoref{fig:result-overview}. Still, several open questions remain. 
% First,~\tclub is polynomial time solvable in interval graphs \cite{Sch09}. 
First, there are obviously parameters for which the parameterized
complexity is still open. For example, is~\tclub parameterized by
\emph{distance to interval graphs} or by \emph{distance to 2-club cluster graphs} in XP or even~fixed-parameter
tractable? 
% Additionally, \autoref{thm:NP-hard-independent-set} which implies that \tclub is NP-hard on $P_7$-free graphs together with \autoref{thm:cographs} which shows fixed-parameter tractability for \emph{distance to $P_4$-free graphs} motivate the question whether \tclub parameterized by \emph{distance to~$P_5$/$P_6$-free graphs} is fixed-parameter tractable.
In this context, also parameter combinations could be of interest. Clearly a complete investigation of the parameter space is infeasible. Hence, one should focus only on practically relevant parameter combinations. One example could be the following question that is left open by the hardness results for \emph{$h$-index} and \emph{degeneracy}. Is~\tclub also W[1]-hard with respect to the parameter~\emph{$h$-index} if the input graph has constant \emph{degeneracy}?
Second, it remains open whether
there is a polynomial kernel for the parameter \emph{distance
to clique} or to identify further non-trivial structural parameters for which polynomial kernels exist. Third, for many of the presented fixed-parameter
tractability results it would be interesting to either improve the running
times or to obtain tight lower bounds. For example, is it possible to solve \tclub parameterized
by \emph{distance to clique} in~$\delta^k\cdot n^{O(1)}$ time for
some~$\delta<2$? Similarly, is it possible to solve \tclub
parameterized by \emph{vertex cover} in~$2^{o(2^k)}\cdot n^{O(1)}$ time?
An answer to the latter question could be a first step towards
improving the (also doubly exponential) running time of the algorithms
for the parameters \emph{treewidth} or \emph{distance to cographs}.
Finally, it would be interesting to see which results carry over to
\textsc{3-Club}~\cite{Pasu08,BM11} or to the related~\textsc{$2$-Clique} problem~\cite{BBT05}. 

{\footnotesize			
 \setlength{\bibsep}{2pt}
\bibliographystyle{abbrvnat}
\bibliography{bibliography}
}
% \appendix
% \newpage
% \section{Proofs}
% 
% \appendixProofText
\end{document}